\documentclass[10pt,twocolumn,twoside]{IEEEtran}

\IEEEoverridecommandlockouts                              	

\usepackage{amssymb,amsmath,amsfonts,amsthm}
\usepackage{algorithm,algorithmic}
\usepackage{cite}
\usepackage{float}
\usepackage{epsfig}
\usepackage{graphicx}
\usepackage{graphics}
\usepackage{subfigure}
\usepackage{url}
\usepackage{verbatim}
\usepackage{xspace}
\usepackage[usenames,dvipsnames]{color}
\usepackage{setspace}

\floatstyle{ruled}
\newfloat{model}{H}{mod}
\floatname{model}{\footnotesize Model}
\newfloat{notatio}{H}{not}
\floatname{notatio}{\footnotesize Notation}

\usepackage{url,changebar,bm,xspace,dsfont}
\let\mathbb=\mathds 

\newtheorem{theorem}{Theorem}

\newtheorem{defn}{Definition}
\newtheorem{prop}{Proposition}

\newtheorem{remark}{Remark}

\ifCLASSINFOpdf
 \else
\fi

\usepackage{graphics} 
\usepackage{epsfig} 
\usepackage{algorithm,algorithmic}

\usepackage{amsmath,amsfonts,amssymb,amscd}
\usepackage{capt-of}
\usepackage{color}
\usepackage{cite}
\usepackage{verbatim}

\begin{document}

\title{\LARGE \bf
Distributed Integer Balancing under Weight Constraints \\ in the Presence of Transmission Delays and Packet Drops}

\author{Apostolos~I.~Rikos,~\IEEEmembership{Member,~IEEE}
	\thanks{Apostolos~I.~Rikos is with the Department of Electrical and Computer Engineering at the University of Cyprus, Nicosia, Cyprus. E-mail:{\tt~arikos01@ucy.ac.cy}.}, 
Christoforos~N.~Hadjicostis,~\IEEEmembership{Senior Member,~IEEE} 
	\thanks{Christoforos~N.~Hadjicostis is with the Department of Electrical and Computer Engineering at the University of Cyprus, Nicosia, Cyprus, and also with the Department of Electrical and Computer Engineering at the University of Illinois, Urbana-Champaign, IL, USA. E-mail:{\tt~chadjic@ucy.ac.cy}.}
	\thanks{Parts of the results for distributed integer weight balancing under interval constraints in the presence of packet drops appear in \cite{2017:RikosHadj}. The present version of the paper includes complete proofs for convergence and proposes extensions to handle transmission delays over the communication links (not addressed in \cite{2017:RikosHadj}).}
}

\maketitle
\thispagestyle{empty}
\pagestyle{empty}

\providecommand{\keywords}[1]{\textbf{\textit{Index terms---}} #1}

\begin{abstract}

We consider the distributed weight balancing problem in networks of nodes that are interconnected via directed edges, each of which is able to admit a positive integer weight within a certain interval, captured by individual lower and upper limits. 
A digraph with positive integer weights on its (directed) edges is weight-balanced if, for each node, the sum of the weights of the incoming edges equals the sum of the weights of the outgoing edges. 
In this work, we develop a distributed iterative algorithm which solves the integer weight balancing problem in the presence of arbitrary (time-varying and inhomogeneous) time delays that might affect
transmissions at particular links.
We assume that communication between neighboring nodes is bidirectional, but unreliable since it may be affected from bounded or unbounded delays (packet drops), independently between different links and link directions.
We show that, even when communication links are affected from bounded delays or occasional packet drops (but not permanent communication link failures), the proposed distributed algorithm allows the nodes to converge to a set of weight values that solves the integer weight balancing problem, after a finite number of iterations with probability one, as long as the necessary and sufficient circulation conditions on the lower and upper edge weight limits are satisfied. 
Finally, we provide examples to illustrate the operation and performance of the proposed algorithms.

\end{abstract}


\IEEEpeerreviewmaketitle

\vspace{-0.3cm}
%
%
%
%

\section{INTRODUCTION}\label{intro}

A distributed system or network consists of a set of components (nodes) that can share information with neighboring components via connection links (edges), forming a generally directed interconnection topology (digraph). 
The digraphs that describe the communication and/or physical topology typically prove to be of vital importance for the effectiveness of distributed strategies in performing various tasks \cite{2003:jadbabaie_coordination, 2007:olfati-saber_consensus, 2008:RenBeard, 2018:BOOK_Hadj}.

A weighted digraph is a digraph in which each edge is associated with a real or integer value called the edge weight. 
Similarly, a flow network (also known as a transportation network) is a digraph where each edge receives a flow that typically cannot exceed a given capacity (or, more generally, has to lie within upper and lower limits). 
A weighted digraph (or flow network) is \textit{weight-balanced} or \textit{balanced} if, for each of its nodes, the sum of the weights of the edges outgoing from the node is equal to the sum of the weights of the edges incoming to the node.

The problem we deal with in this paper can be viewed as the problem of weight/flow balancing under integer weight/flow constraints on each edge of a given digraph \cite{2014:RikosHadj}, or the problem of producing a feasible circulation in a directed graph with upper and lower flow constraints \cite{1986:Papadimitriou}.
Furthermore, it can also be seen as a particular case of the standard network flow problem (see, e.g., \cite{2010:Fulkerson}), where there is a cost associated to the flow on each link, and the objective is to minimize the total cost subject to balancing constraints on the flows.

Weight-balanced digraphs find numerous applications in distributed adaptive control and synchronization in complex networks. 
Examples of applications where balance plays a key role include modeling of flocking behavior \cite{2003:jadbabaie_coordination}, network adaptation strategies based on the use of continuous second order models \cite{2010:DeLellis}, prediction of distribution matrices for telephone traffic \cite{1998:Bertsekas}, distributed adaptive strategies to tune the coupling weights of a network based on local information of node dynamics \cite{2012:YuDeLellis}, and design of cut-balanced networks for consensus seeking systems \cite{2013:Julien_Tsitsiklis}. 
Weight/flow balance is also closely related to weights/flows that form a doubly stochastic matrix \cite{2012:CortesJournal}, which find applications in multicomponent systems (such as sensor networks) where one is interested in distributively averaging measurements at each component. 
Asymptotic consensus to the real average \cite{1984:Tsitsiklis} or the quantized average \cite{2016:ChamieBasar} of the initial values is guaranteed if the weights used in the linear iteration form a doubly stochastic matrix. 
In particular, the distributed average consensus problem has received significant attention from the computer science community \cite{1996:Lynch} and the control community \cite{2004:XiaoBoyd} due to its applicability to diverse areas, including multi-agent systems, distributed estimation and tracking \cite{2008:CarliChiusoSchenatoZampieri}, and distributed optimization \cite{2010:Nedic}.
A review of recent approaches to distributed average consensus (and its applications to various settings) can be found in \cite{2018:BOOK_Hadj}.

Recently, quite a few works have dealt with the problem of balancing a strongly connected digraph with either real or integer weights/flows. 
For example, \cite{2009:Cortes, 2018:BOOK_Hadj, 2014:RikosHadj, 2013:Priolo, RikosHadjDelays_2017} deal with distributed algorithms for weight/flow balancing when the nonnegative weights on each edge are otherwise unconstrained (in terms of the values they admit), \cite{2016:RikosHadj, 2016:HadjAlej} deal with the problem of weight/flow balancing assuming timely and reliable exchange of information between nodes, and \cite{2017:RikosHadj_Allerton} deals with weight/flow balancing when the nonnegative weights on each edge are constrained and admit real values (resulting to asymprotic convergence) in the presence of unreliable communication links. 

In this paper, we investigate the problem of integer weight/flow balancing in a multi-component system under a directed interconnection topology, where the flow/weight on each edge consists of lower and upper constraints (in terms of the values it admits) in the presence of bounded delays or unbounded delays (packet drops) in the communication links. 
We consider a fixed topology (digraph) and we devise a protocol, based on our previous work in \cite{2017:RikosHadj}, where each node updates its state by combining the available (possibly delayed) weight information received by its in-neighbors. 
We establish that the proposed balancing algorithm reaches, after a finite number of steps, a set of weights that form a weight-balanced digraph despite the presence of arbitrary but bounded delays in the communication links. 
When packet drops (i.e., infinite delays) are present over the communication links, we propose a modified version of the algorithm that is shown to converge to  a set of weights that form a balanced graph after a finite number of iterations (with probability one). 
In both cases, we argue that the proposed algorithm reaches a solution as long as such as set of weights exists.

The remainder of this paper is organized as follows. 
In Section~\ref{preliminaries} the notation used throughout the paper is provided, along with background on graph theory and the problem formulation. 
In Section~\ref{CircConditions} we present the conditions for the existence of a set of integer weights (within the interval constraints) that balance a weighted digraph.
In Section~\ref{algorithm_delays} we present the distributed algorithm which achieves integer weight-balancing in the presence of bounded delays after a finite number of iterations. 
In Section~\ref{algorithm_packetDr}, we analyze the case of unbounded delays (packet drops) in the communication links and we present a distributed algorithm which achieves integer weight-balancing after a finite number of iterations with probability one. 
Finally, in Section~\ref{results} we present simulation results and comparisons, and we conclude in Section~\ref{future} with a brief summary and remarks about our future work.

%
%
%
%
\section{NOTATION AND BACKGROUND}\label{preliminaries}

Matrices are denoted by capital letters. The sets of real, integer, natural and nonnegative integer numbers are denoted by $ \mathbb{R}, \mathbb{Z}$, $\mathbb{N}$ and $\mathbb{N}_0$ respectively.

\subsection{Graph-Theoretic Notions}

A distributed system whose components can exchange certain quantities of interest via (possibly directed) links, can conveniently be captured by a digraph (directed graph). 
A digraph of order $n$ ($n \geq 2$), is defined as $\mathcal{G}_d = (\mathcal{V}, \mathcal{E})$, where $\mathcal{V} =  \{v_1, v_2, \dots, v_n\}$ is the set of nodes and $\mathcal{E} \subseteq \mathcal{V} \times \mathcal{V} - \{ (v_j,v_j)$ $|$ $v_j \in \mathcal{V} \}$ is the set of edges. 
A directed edge from node $v_i$ to node $v_j$ is denoted by $(v_j, v_i) \in \mathcal{E}$, and indicates that a nonnegative flow of mass from node $v_i$ to node $v_j$ is possible. We will refer to the digraph $\mathcal{G}_d$ as the {\em topology}.

A digraph is called \textit{strongly connected} if for each pair of vertices $v_j, v_i \in \mathcal{V}$, $v_j \neq v_i$, there exists a directed \textit{path} from $v_i$ to $v_j$, i.e., we can find a sequence of vertices $v_i \equiv v_{l_0},v_{l_1}, \dots, v_{l_t} \equiv v_j$ such that $(v_{l_{\tau+1}},v_{l_{\tau}}) \in \mathcal{E}$ for $ \tau = 0, 1, \dots , t-1$. 
All nodes that can have flows to node $v_j$ directly are said to be in-neighbors of node $v_j$ while the nodes that receive flows from node $v_j$ comprise its out-neighbors.
The in- and out-neighbors of node $v_j$ are nodes in the set $\mathcal{N}_j^- = \{v_i \in \mathcal{V} \; | \; (v_j, v_i) \in \mathcal{E} \}$ and $\mathcal{N}_j^+ = \{v_l \in \mathcal{V} \; | \; (v_l, v_j) \in \mathcal{E} \}$ respectively, where the cardinality of $\mathcal{N}_j^-$ is called the \textit{in-degree}  of $v_j$ (denoted by $\mathcal{D}_j^-$) and the cardinality of $\mathcal{N}_j^+$ is called the \textit{out-degree} of $v_j$ (denoted by $\mathcal{D}_j^+$). 
We let $\mathcal{N}_j = \mathcal{N}^+_j \cup  \mathcal{N}^-_j$ denote the {\em neighbors} of node $v_j$, and $\mathcal{D}_j = \mathcal{D}^+_j + \mathcal{D}^-_j$ denote the {\em total degree} of node $v_j$. 
Also, $ \mathcal{E}^-_j  =  \{ (v_j, v_i) \ | \ v_i \in \mathcal{N}^-_j \} $ ($ \mathcal{E}^+_j  =  \{ (v_l, v_j) \ | \ v_l \in \mathcal{N}^+_j \} $) denotes the incoming (outgoing) edges to (from) node $v_j$. 
Note that $ | \mathcal{E}^-_j \cup \mathcal{E}^+_j | = \mathcal{D}^+_j + \mathcal{D}^-_j = \mathcal{D}_j$, where $\mathcal{D}_j$ is the total degree of node $v_j$. 

We assume that node $v_j$ assigns a ``unique order'' in the set $\{0,1,..., \mathcal{D}_j-1\}$ to each of its outgoing and incoming edges. 
The order of edge $(v_l, v_j)$ (or edge ($v_j$, $v_i$)) is denoted by $P_{lj}$ (or $P_{ji}$) (such that $\{P_{lj} \; | \; v_l \in \mathcal{N}^+_j\} \cup \{P_{ji} \; | \; v_i \in \mathcal{N}^-_j\} = \{0,1,..., \mathcal{D}_j-1\}$) and will be used later on as a way of allowing node $v_j$ to make changes to its outgoing and incoming edge flows in a predetermined order. 
Note that the ``unique order'' is cyclic in the sense that every time a node attempts to change the flows of its incoming/outgoing edges, it continues from the edge it stopped the previous time according to the predetermined order, starting from the beginning if it has changed the values of every incoming and outgoing edge.

We assume that a pair of nodes $v_j$ and $v_i$ that are connected by an edge in the digraph $\mathcal{G}_d$ (i.e., $(v_j, v_i) \in \mathcal{E}$ and/or $(v_i, v_j) \in \mathcal{E}$) can exchange information among themselves (in both directions). 
In other words, the {\em communication topology} is captured by the undirected graph $\mathcal{G}_u = (\mathcal{V}, \mathcal{E}_u)$ that corresponds to the given directed graph $\mathcal{G}_d = (\mathcal{V}, \mathcal{E})$, where $ \mathcal{E}_u = \cup_{(v_j, v_i) \in \mathcal{E}} \{ (v_j, v_i), (v_i, v_j) \} $.


\subsection{Flow/Weight Balancing}

Given a digraph $\mathcal{G}_d = (\mathcal{V}, \mathcal{E})$ we aim to assign positive integer flows $f_{ji} \in \mathbb{N}$ to each edge $(v_j, v_i) \in \mathcal{E}$. 
In this paper, these flows will be restricted to lie in an interval $[l_{ji}, u_{ji}]$, i.e., $0 < l_{ji} \leq f_{ji} \leq u_{ji}$ and $f_{ji} \in \mathbb{N}$, for every $(v_j, v_i) \in \mathcal{E}$. 
We will also use matrix notation and denote (respectively) the integer flow, {\em perceived} integer flow\footnote{The perceived integer flow $f_{ji}^{(p)}$ will be used to denote the flow that node $v_j$ perceives on link $(v_j, v_i)$; due transmission delays or packet drops the flow perceived by node $v_j$ might be different from the actual flow $f_{ji}$ assigned by node $v_i$ (our convention is that the true flow on edge $(v_j,v_i)$ is assigned by node $v_i$).}, lower limit, and upper limit matrices by the $n \times n$ matrices $F = [ f_{ji} ]$, $F_p = [ f^{(p)}_{ji} ]$, $L=[l_{ji}]$, and $U=[u_{ji}]$, where $F(j,i)=f_{ji}$, $F_p(j,i)=f^{(p)}_{ji}$, $L(j,i)=l_{ji}$, $U(j,i)=u_{ji}$, for every $(v_j, v_i) \in \mathcal{E}$ (obviously, $f_{ji} = f^{(p)}_{ji} = l_{ji} = u_{ji} = 0$ when $(v_j,v_i) \notin \mathcal{E}$).

\begin{defn}\label{DEFnodebalance}
Given a digraph $\mathcal{G}_d=(\mathcal{V},\mathcal{E})$ of order $n$ along with an integer flow assignment $F=[f_{ji}] $, the total \textit{in-flow} of node $v_j$ is defined as $f_j^- = \sum_{v_i \in \mathcal{N}_j^-} f_{ji}$, the total \textit{out-flow} of node $v_j$ is defined as $f_j^+ = \sum_{v_l \in \mathcal{N}_j^+} f_{lj}$ and the {\em flow balance} $b_j$ of node $v_j$ is $b_j = f_j^- - f_j^+$.
\end{defn}

\begin{defn}
\label{DEFpercnodebalance}
Given a digraph $\mathcal{G}_d=(\mathcal{V},\mathcal{E})$ of order $n$, along with an integer flow assignment $F=[f_{ji}]$ and a perceived flow assignment $F_p=[f^{(p)}_{ji}]$, the total {\em perceived} in-flow $f_j^{-(p)}$ of node $v_j$ is $f_j^{-(p)} = \sum_{v_i \in \mathcal{N}_j^-} f^{(p)}_{ji}$ while the {\em perceived flow balance} $b^{(p)}_j$ of node $v_j$ is $b^{(p)}_j = f_j^{-(p)} - f_j^+$.
\end{defn}

\begin{defn}
\label{defn:totalim}
Given a digraph $\mathcal{G}_d=(\mathcal{V},\mathcal{E})$ of order $n$, along with an integer flow assignment $F=[f_{ji}]$, the \textit{total imbalance} (or {\em absolute imbalance}) of digraph $\mathcal{G}_d$ is defined as $\varepsilon = \sum_{j=1}^{n} \vert b_j \vert$, while the perceived total imbalance of digraph $\mathcal{G}_d$ is defined as $\varepsilon^{(p)} = \sum_{j=1}^{n} \vert b^{(p)}_j \vert$. The digraph $\mathcal{G}_d$ is called flow-balanced if its \textit{total imbalance} is zero.
\end{defn}

\begin{remark}
Note here that the integer flow $f_{lj}$ on edge $(v_l, v_j) \in \mathcal{E}$ is assigned by node $v_j$. 
Thus, node $v_j$ has access to the true flow $f_{lj}$ of edge $(v_l, v_j)$ while node $v_l$ has access to a {\em perceived} flow $f^{(p)}_{lj}$, which will be equal to $f_{lj}$ if node $v_j$ is able to successfully communicate with node $v_l$. 
This means that node $v_l$ can only calculate its {\em perceived flow balance} $b^{(p)}_l$ at each iteration $k$ and it has no access to the total (or perceived total) imbalance of the digraph $\mathcal{G}_d$.
\end{remark}

\subsection{Modeling Time Delays and Packet Drops}\label{model_delays}

We assume that a transmission from node $v_j$ to node $v_l$ at time step $k$ undergoes an \textit{a priori unknown} delay $\tau^{(j)}_{lj}[k]$ while, we consider both bounded delays and unbounded delays (packet drops). 
For bounded delays, we assume that $\tau^{(j)}_{lj}[k]$ is an integer that satisfies $0 \leq \tau^{(j)}_{lj}[k] \leq \overline{\tau}_{lj} \leq \infty$ where the maximum delay is denoted by $\overline{\tau} = \max_{(v_l,v_j) \in \mathcal{E}}\overline{\tau}_{lj}$. 
In the weight balancing setting we consider that node $v_j$ is in charge of assigning the {\em actual} flow $f_{lj}[k]$ to each link $(v_l,v_j)$, and then transmits to node $v_l$ the amount of change $c^{(j)}_{lj}[k]$ it desires at time step $k$. 
Under the above delay model (which also assumes bidirectional communication), node $v_l$ ($v_j$) receives the change amount $c^{(j)}_{lj}[k]$ ($c^{(l)}_{lj}[k]$), required by node $v_j$ ($v_l$) over the actual (perceived) flow $f_{lj}[k]$ ($f^{(p)}_{lj}[k]$), at time step $k + \tau^{(j)}_{lj}[k]$ ($k + \tau_{lj}^{(l)}[k]$).

To handle the case when a transmission from node $v_j$ to node $v_l$ at time step $k$ undergoes an \textit{a priori unknown unbounded} delay, we assume that each particular edge may drop packets with some non-total probability. 
We assume independence between packet drops at different time steps or different links (or even different directions of the same link), so that, we can model a packet drop via a Bernoulli random variable:
\begin{equation}\label{dropsmodel}
Pr\{ x_k(j,i)=m \} = \left\{ \begin{array}{ll}
         q_{ji}, & \mbox{if $m = 0$,}\\
         1 - q_{ji}, & \mbox{if $m = 1$,}\end{array} \right.
\end{equation}
where $x_k(j,i)=1$ if the transmission from node $v_i$ to node $v_j$ at time step $k$ is successful.

\subsection{Problem formulation}

We are given a strongly connected digraph $\mathcal{G}_d = (\mathcal{V}, \mathcal{E})$, as well as lower and upper limits $l_{ji}$ and $u_{ji}$ ($0 < l_{ji} \leq u_{ji}$, where $l_{ji}, u_{ji} \in \mathbb{R}$) on each each edge $(v_j, v_i) \in \mathcal{E}$. 
Considering that link transmissions undergo arbitrary, bounded (or unbounded) delays, we want to develop a distributed algorithm that allows the nodes to iteratively adjust the integer flows on their edges so that they eventually obtain a set of integer flows $\{ f_{ji} \; | \; (v_j, v_i) \in \mathcal{E} \}$ that satisfy the following:
\begin{enumerate}
\item $ f_{ji} \in \mathbb{N}$ for each edge $(v_j,v_i) \in \mathcal{E}$;
\item $l_{ji} \leq f_{ji} \leq u_{ji}$ for each edge $(v_j,v_i) \in \mathcal{E}$;
\item $f_j^+ = f_j^- = f_j^{-(p)}$ for each $v_j \in \mathcal{V}$.
\end{enumerate}
The distributed algorithm needs to respect the communication constraints imposed by the undirected graph $\mathcal{G}_u$ that corresponds to the given directed graph~$\mathcal{G}_d$. 
Specifically, the {\em communication topology} is captured by the undirected graph $\mathcal{G}_u = (\mathcal{V}, \mathcal{E}_u)$ that corresponds to a given directed graph $\mathcal{G}_d = (\mathcal{V}, \mathcal{E})$, where $\mathcal{E}_u = \cup_{(v_j, v_i) \in \mathcal{E}} \{ (v_j, v_i), (v_i, v_j) \}$.

\begin{remark}
One of the main differences of the work in this paper with the works in \cite{2010:Cortes, 2009:Cortes, 2014:RikosHadj, 2013:Priolo, 2012:Rikos, RikosHadjDelays_2017} is that the algorithm developed in this paper requires a bidirectional communication topology, whereas most of the aforementioned works assume a communication topology that matches the flow (physical) topology. 
We should point out that direct application of these earlier algorithms to the problem that is of interest in this paper will generally fail (because flows are restricted to lie within lower and upper limits). 
Also, note that there are many applications where the physical topology is directed but the communication topology is bidirectional. 
One such example is the traffic network that was mentioned earlier; it is represented by a digraph, in which unidirectional or bidirectional edges (possibly capacity constrained) capture, respectively, one-way or two-way streets, and where nodes capture intersections. Traffic lights typically sit at these intersections and aim to control traffic flow; even though traffic lights may be constrained in terms of how they divert flow (depending on the constraints of the traffic network), communication between neighboring traffic lights can be bidirectional. In other words, there are two graphs: the directed (physical) graph representing the actual traffic flow over streets/edges and the likely undirected (cyber or communication) graph representing the communication capability between nodes in the graph. In applications like the traffic network mentioned above, the algorithms proposed here are directly applicable. More generally, in many applications, the communication topology may not necessarily match the physical one; in our future work, we plan to enhance the algorithm proposed here to allow for different communication topologies (including the one that matches the physical topology).
\end{remark}

\section{NECESSARY AND SUFFICIENT CONDITIONS: INTEGER CIRCULATION CONDITIONS}\label{CircConditions}

When edge weights are restricted to be integers, the theorem below (a variation of the well known circulation conditions) characterizes the necessary and sufficient conditions (e.g., see Theorem $3.1$ in \cite{2010:Fulkerson}) for the existence of a set of integer flows that satisfy interval constraints and balance constraints.

\begin{theorem}\label{IntTheoremCirc}
Consider a strongly connected digraph $\mathcal{G}_d=(\mathcal{V}, \mathcal{E})$, with lower and upper bounds $l_{ji}$ and $u_{ji}$ (where $0 < l_{ji} \leq u_{ji}$) on each edge $(v_j, v_i) \in \mathcal{E}$. 
The \textit{necessary} and \textit{sufficient} conditions for the existence of a set of integer flows $\{ f_{ji} \in \mathbb{N} \ | \; (v_j, v_i) \in \mathcal{E} \}$ that satisfy
\begin{enumerate}
\item {\em Interval constraints:} $0 < l_{ji} \leq f_{ji} \leq u_{ji}$ for each edge $(v_j,v_i) \in \mathcal{E}$, and
\item {\em Balance constraints:} $f_j^+ = f_j^-$ for every $v_j \in \mathcal{V}$,
\end{enumerate}
are the following: 
\begin{enumerate}
\item[(i)] for every $(v_j,v_i) \in \mathcal{E}$, we have $\lceil l_{ji} \rceil \leq \lfloor u_{ji} \rfloor,$ and
\item[(ii)] for each $\mathcal{S}$, $\mathcal{S} \subset \mathcal{V}$, we have
\begin{equation}
\label{EQnsconditions}
\sum_{(v_j, v_i) \in \mathcal{E}^-_\mathcal{S}} \lceil l_{ji} \rceil \leq \sum_{(v_l, v_j) \in \mathcal{E}^+_\mathcal{S}} \lfloor u_{lj} \rfloor \; ,
\end{equation}
\end{enumerate}
where 
\begin{eqnarray}
\mathcal{E}^-_\mathcal{S} & = & \{ (v_j, v_i) \in \mathcal{E} \; | \; v_j \in \mathcal{S}, \; v_i \in \mathcal{V}-\mathcal{S} \} \; , \label{REALEQinS}\\
\mathcal{E}^+_\mathcal{S} & = & \{ (v_l, v_j) \in \mathcal{E} \; | \; v_j \in \mathcal{S}, \; v_l \in \mathcal{V}-\mathcal{S} \} \label{REALEQoutS}\; .
\end{eqnarray}
\end{theorem}

\begin{remark}
Note that Theorem~\ref{IntTheoremCirc} effectively requires $\mathcal{G}_d$ to be strongly connected or a pure collection of strongly connected sub-digraphs.
The necessity of the conditions described in Theorem~\ref{IntTheoremCirc} follows from the conditions in \cite{2010:Fulkerson}: when flows are restricted to be integers, the effective interval of $f_{ji}$ is the interval $[\lceil l_{ji} \rceil, \lfloor u_{ji} \rfloor]$ and clearly has to be non-empty for each $(v_j, v_i) \in {\mathcal V}$ (condition (i) above).
\end{remark}

%
%
%
%
\section{INTEGER FLOW BALANCING ALGORITHM WITH TIME DELAYS}
\label{algorithm_delays}

In this section we provide an overview of the distributed flow algorithm operation; the formal description of the algorithm is provided in Algorithm~\ref{alg1_delays}. 
The algorithm is iterative and operates by having, at each iteration, nodes with \textit{positive} \textit{perceived} flow balance attempt to change the integer flows on both their incoming and/or outgoing edges so that they become flow balanced. 
We first describe the distributed iterative algorithm operations and we establish that, if the necessary and sufficient integer circulation conditions in Theorem~\ref{IntTheoremCirc} are satisfied, the algorithm completes after a finite number of iterations.

\textbf{Initialization.} At initialization, each node is aware of the feasible flow interval on each of its incoming and outgoing edges, i.e., node $v_j$ is aware of $l_{ji}, u_{ji}$ for each $v_i \in \mathcal{N}^-_j$ and $l_{lj}, u_{lj}$ for each $v_l \in \mathcal{N}^+_j$. 
Furthermore, the flows are initialized at the ceiling of the lower bound of the feasible interval, i.e., $f_{ji}[0] = \lceil l_{ji} \rceil$.
This initialization is always feasible but not critical and could be any integer value in the feasible flow interval $[l_{ji}, u_{ji}]$ (according to Theorem~\ref{IntTheoremCirc} an integer always exists in the interval $[l_{ji}, u_{ji}]$).
Also, each node $v_j$ chooses a unique order $P_{lj}^{(j)}$ and $P_{ji}^{(j)}$ for its outgoing links $(v_l,v_j)$ and incoming links $(v_j,v_i)$ respectively, such that $\{P_{lj}^{(j)} \; | \; v_l \in \mathcal{N}^+_j\} \cup \{P_{ji}^{(j)} \; | \; v_i \in \mathcal{N}^-_j\} = \{0,1,..., \mathcal{D}_j-1\}$. 

\textbf{Iteration.} At each iteration $k \geq 0$, node $v_j$ is aware of the {\em perceived} integer flows on its incoming edges $\{ f^{(p)}_{ji}[k] \; | \: v_i \in \mathcal{N}^-_j \}$ and the (actual) flows on its outgoing edges $\{ f_{lj}[k] \; | \: v_l \in \mathcal{N}^+_j \}$, which allows it to calculate its {\em perceived} flow balance $b^{(p)}_j[k]$ according to Definition~\ref{DEFpercnodebalance}.

\noindent
{\em A. Selecting Desirable Flows.} Each node $v_j$ with positive {\em perceived} flow balance $b^{(p)}_j[k] > 0$ attempts to subtract $1$ (one unit of flow) from the flows on its incoming edges $\{ f_{ji}[k] \; | \; v_i \in \mathcal{N}^-_j \}$ and add $1$ (one unit of flow) to the flows of its outgoing edges $\{ f_{lj}[k] \; | \; v_l \in \mathcal{N}^+_j \}$, one at a time by following the predetermined order (chosen at initialization) in a round-robin fashion, until its perceived flow balance $b^{(p)}_j[k+1]$ becomes zero (at least if no other changes are inflicted on the flows).  
If an outgoing (incoming) edge has reached its max (min) value (according to the feasible interval on that particular edge), then its flow does not change and node $v_j$ proceeds in changing the flow of the ensuing edge, according to the predetermined order. 
Note here that no attempt to change flows is made if node $v_j$ has negative or zero perceived flow balance. 
The next time node $v_j$ needs to change the flows of its incoming/outgoing edges, it will continue from the edge it stopped the previous time and cycle through the edge weights in a round-robin fashion according to the ordering chosen at initialization. 
The desired flow change by node $v_j$ on edge $(v_j,v_i) \in \mathcal{E}$ at iteration $k$ will be denoted by $c_{ji}^{(j)}[k]$; similarly, the desired flow change by node $v_j$ on edge $(v_l,v_j) \in \mathcal{E}$ at iteration $k$ will be denoted by $c_{lj}^{(j)}[k]$.

\noindent
{\em B. Exchanging Desirable Flows.} 
Once each node $v_j$ with positive {\em perceived} flow balance calculates the desirable flow change for each incoming $\{ c^{(j)}_{ji}[k] \; | \; v_i \in \mathcal{N}^-_j \}$ and outgoing $\{ c^{(j)}_{lj}[k] \; | \; v_l \in \mathcal{N}^+_j \}$ flow, it does the following steps in sequence: 

\noindent
1) It transmits the desirable flow change $c^{(j)}_{ji}[k]$ ($c^{(j)}_{lj}[k]$) to each in- (out-) neighbor $v_i$ ($v_l$).

\noindent
2) It receives the (possibly delayed) desired flow changes $\overline{c}^{(i)}_{ji}[k]$ ($\overline{c}^{(l)}_{lj}[k]$) from each in- (out-) neighbor $v_i$ ($v_l$).
From node $v_j$'s perspective, the delayed flow change for link $(v_l,v_j)$, $\forall v_l \in \mathcal{N}^+_j$, at time step $k$ is given by 
\begin{align} \label{delay_equation}
\overline{c}^{(l)}_{lj}[k]= \sum_{k_0 = k - \overline{\tau}}^{k} c^{(l)}_{lj}[k_0], \text{ for which } k_0 + \tau^{(l)}_{lj}[k_0] = k ,  
\end{align}
i.e., $\overline{c}^{(l)}_{lj}[k]$ is the sum of flow changes $c^{(l)}_{lj}$ that were sent from $v_l$ and are seen by node $v_j$ by time step $k$.
If no flow change is received due to time delays, then node $v_j$ assumes that $\overline{c}^{(i)}_{ji}[k] = 0$ ($\overline{c}^{(l)}_{lj}[k] = 0$) for the corresponding incoming (outgoing) edge $(v_j, v_i)$ ($(v_l, v_j)$).

\noindent
3) It calculates its new outgoing ({\em perceived} incoming) flows $f_{lj}[k+1] = f_{lj}[k] + c_{lj}^{(j)}[k] + \overline{c}_{lj}^{(l)}[k] $ ($f_{ji}^{(p)}[k+1] = f_{ji}^{(p)}[k] + c_{ji}^{(j)}[k] + \overline{c}_{ji}^{(i)}[k]$).
Then, the new outgoing ({\em perceived} incoming) flows are adjusted so that the new flow is projected onto the feasible interval $[l_{lj}, u_{lj}]$ ($[l_{ji}, u_{ji}]$) of the corresponding edge.
This (along with all the parameters involved) can be seen in Figure~\ref{nodes_exch_delays}.

\begin{remark}
Since the flow $f_{ji}$ on each edge $(v_j, v_i) \in \mathcal{E}$ affects positively the flow balance $b_j[k]$ of node $v_j$ and negatively the flow balance $b_i[k]$ of node $v_i$, we need to take into account the possibility that both nodes desire a change on the flow simultaneously. 
Thus, the proposed algorithm attempts to coordinate the flow change. 
The challenge, however, is the fact that time delays may occur during transmissions (in either direction) while the nodes are trying to agree on a flow value.
\end{remark}

\begin{figure}[h]
\begin{center}
\includegraphics[width=0.8\columnwidth]{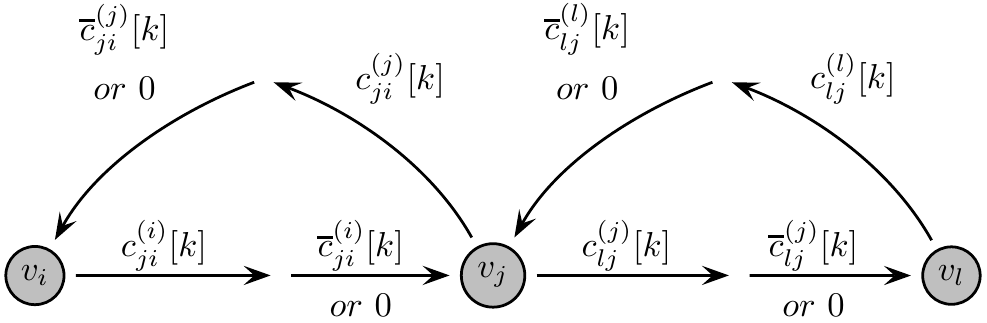}
\caption{Digraph where nodes exchange their desirable flows in the presence of time delays.}
\label{nodes_exch_delays}
\end{center}
\end{figure}


\begin{algorithm}
\caption{Distributed Flow Balancing Algorithm}
\textbf{Input} \\ 1) A strongly connected digraph $\mathcal{G}_d = (\mathcal{V}, \mathcal{E})$ with $n=|\mathcal{V}|$ nodes and $m=|\mathcal{E}|$ edges.\\ 2) $l_{ji},u_{ji}$ for every $(v_j, v_i) \in \mathcal{E}$, such that the circulation conditions in Theorem~\ref{IntTheoremCirc} are satisfied. \\
\textbf{Initialization} \\ Set $k=0$; each node $v_j \in \mathcal{V}$ does:
\\1) It sets the flows on its {\em perceived} incoming and outgoing edge flows as 
$$
f^{(p)}_{ji}[0] = \lceil l_{ji} \rceil, \ \forall  v_i \in \mathcal{N}_j^-,
$$
$$
f_{lj}[0] = \lceil l_{lj} \rceil, \ \forall  v_l \in \mathcal{N}_j^+.
$$
2) It assigns a unique order to its outgoing and incoming edges as $P_{lj}^{(j)}$, for $v_l \in \mathcal{N}_j^+$ or $P_{ji}^{(j)}$, for $v_i \in \mathcal{N}_j^-$ (such that $\{P_{lj}^{(j)} \ | \ v_l \in \mathcal{N}_j^+\} \cup \{P_{ji}^{(j)} \ | \ v_i \in \mathcal{N}_j^-\} = \{0,1,...,\mathcal{D}_j -1 \}$).
\\
\textbf{Iteration} \\ For $k=0,1,2,\dots$, each node $v_j \in \mathcal{V}$ does the following:
\\ 1) It computes its \textit{perceived flow balance} as in Definition~\ref{DEFpercnodebalance}
$$
b^{(p)}_j[k] = \sum_{v_i \in \mathcal{N}_j^-} f_{ji}^{(p)}[k] - \sum_{v_l \in \mathcal{N}_j^+} f_{lj}[k]. 
$$
2) If $b^{(p)}_j[k] > 0$, it increases (decreases) by $1$ the integer flows $f_{lj}[k]$ ($f_{ji}^{(p)}[k]$) of its outgoing (incoming) edges $v_l \in \mathcal{N}_j^+$ ($v_i \in \mathcal{N}_j^-$) one at a time, following the predetermined order $P_{lj}^{(j)}$ ($P_{ji}^{(j)}$) until its flow balance becomes zero (if an edge
has reached its maximum (minimum) value and it cannot be increased (decreased) further, its flow does not change and node $v_j$ proceeds in changing the next one according to the predetermined order). 
Then, it stores the desired change amount for each outgoing edge as $c_{lj}^{(j)}[k]$ and each incoming edge as $c_{ji}^{(j)}[k]$.
\\ 3) If  $b^{(p)}_j[k] > 0$, it transmits the desired flow change $c^{(j)}_{lj}[k]$ ($c^{(j)}_{ji}[k]$)  on each outgoing (incoming) edge.
\\ 4) It receives the (possibly delayed) desired flow change $\overline{c}^{(l)}_{lj}[k]$ ($\overline{c}^{(i)}_{ji}[k]$) from each outgoing (incoming) edge. 
[If no flow change is received due to time delays it assumes $\overline{c}^{(l)}_{lj}[k] = 0$ ($\overline{c}^{(i)}_{ji}[k] = 0$) for the corresponding outgoing (incoming) edge.]
\\ 5) It sets its outgoing flows to be 
$$
f_{lj}[k+1] = f_{lj}[k] + c_{lj}^{(j)}[k] + \overline{c}_{lj}^{(l)}[k] ,
$$
and its new {\em perceived} incoming flows to be 
$$
f_{ji}^{(p)}[k+1] = f_{ji}^{(p)}[k] + c_{ji}^{(j)}[k] + \overline{c}_{ji}^{(i)}[k] .
$$
6) It adjusts the new outgoing flows according to the corresponding upper and lower weight constraints as
$$
f_{lj}[k+1] =\max \{ l_{lj}, \min\{u_{lj}, f_{lj}[k+1]\} \} ,
$$ 
and its new {\em perceived} incoming flows according to the corresponding upper and lower weight constraints as
$$
f_{ji}^{(p)}[k+1] =\max \{ l_{ji}, \min\{u_{ji}, f_{ji}^{(p)}[k+1]\} \} .
$$
7) It repeats (increases $k$ to $k+1$ and goes back to Step~1).
\label{alg1_delays}
\end{algorithm}

\begin{remark}\label{perceivedSmaller_del}
It is important to note here that the total {\em perceived} in-flow $f_j^{-(p)}$ of node $v_j$ might be affected from possible time delays at Step~4 of Algorithm~\ref{alg1_delays}. 
Specifically, if transmissions are are such that node $v_j$ receives no information about the change desired by its in-neighbor $v_i$ on the flow $f_{ji}$, then $v_j$ sets $f_{ji}^{(p)}[k+1] = f_{ji}^{(p)}[k] + c_{ji}^{(j)}[k]$ where $c_{ji}^{(j)}[k] < 0$.
Since nodes only attempt to make changes on the flows if their perceived balance is positive, node $v_i$ will only attempt to increase the flow $f_{ji}[k]$ of edge $(v_j, v_i)$,
which means that during the execution of Algorithm~\ref{alg1_delays} we have $f_{ji}^{(p)}[k] \leq f_{ji}[k]$ for each edge $(v_j,v_i) \in \mathcal{E}$, at each time step $k$. 
For this reason, we also have that the perceived balance $b_j^{(p)}[k]$ of node $v_j$ at iteration $k$ is always smaller or equal to its actual balance $b_j[k]$ (i.e., $b_{j}^{(p)}[k] \leq b_{j}[k]$). 
\end{remark}

\begin{remark}
According to the circulation conditions in Theorem~\ref{CircConditions}, each node $v_j \in \mathcal{V}$ with positive {\em perceived} flow balance at iteration $k$ ($b^{(p)}_j[k] > 0$) will always be able to calculate a flow assignment for its incoming and outgoing edge flows so that its {\em perceived} flow balance becomes zero (at least if no other changes are inflicted on the flows of its incoming or outgoing links). 
This can easily be seen by taking the set $\mathcal{S}$ to be $\{ v_j \}$, and realizing that the circulation conditions allow a flow assignment that is balanced. 
This means that the selection of desirable flows in Algorithm~\ref{alg1_delays} is always {\em feasible}.
\end{remark}

%

%
%
%
%
\vspace{-0.6cm}
\subsection{Proof of Algorithm Completion}\label{convergence_delays}

In this section, we show that, as long as the circulation conditions in Theorem~\ref{IntTheoremCirc} hold, then the total imbalance $\varepsilon[k]$ in Definition~\ref{defn:totalim} goes to zero after a finite number of iterations of Algorithm~\ref{alg1_delays}. 
This implies that the flow balance $b_j[k]$ for each node $v_j \in \mathcal{V}$ goes to zero after a finite number of iterations.  
We will argue that the perceived flow balance $b_j^{(p)}$ also goes to zero and that the flows and perceived flows on each edge $(v_j, v_i) \in \mathcal{E}$ stabilize to the same  integer value $f_{ji}^*$ (where $f_{ji}^* \in \mathbb{N}_0$) within the given lower and upper limits, i.e., $0 < l_{ji} \leq f^*_{ji} \leq u_{ji}$ for all $(v_j, v_i) \in \mathcal{E}$. 
As in \cite{2016:RikosHadj} we assume that $l_{ji} \geq 1$ for each edge $(v_j, v_i) \in \mathcal{G}$.

\noindent
We begin by establishing some basic propositions, which rely on the strong connectivity of the network. Complete proofs of Proposition 1 can be found in \cite{2016:HadjAlej}.

\begin{prop}
\label{PROP1_del}
Consider the problem formulation described in Section~\ref{preliminaries}. 
At each iteration $k$ during the execution of Algorithm~\ref{alg1_delays}, it holds that
\begin{enumerate}
\item For any subset of nodes $\mathcal{S} \subset \mathcal{V}$, let $\mathcal{E}^-_\mathcal{S}$ and $\mathcal{E}^+_\mathcal{S}$ be defined by (\ref{REALEQinS}) and (\ref{REALEQoutS}) respectively. Then,
$$
\sum_{v_j \in \mathcal{S}} b_j[k] = \sum_{(v_j, v_i) \in \mathcal{E}^-_\mathcal{S}} f_{ji}[k] - \sum_{(v_l, v_j) \in \mathcal{E}^+_\mathcal{S}} f_{lj}[k] \; ;
$$
\item $\sum_{j=1}^n b_j[k] = 0$;
\item $\varepsilon[k] = 2 \sum_{v_j \in \mathcal{V}^-[k]} \vert b_j[k] \vert$ where $\mathcal{V}^-[k] = \{ v_j \in \mathcal{V} \; | \; b_j[k] < 0 \}$.
\end{enumerate}
\end{prop}

\begin{prop}
\label{PROP2_del}
Consider the problem formulation described in Section~\ref{preliminaries}. 
Let $\mathcal{V}^-[k] \subset \mathcal{V}$ be the set of nodes with negative flow balance at iteration $k$, i.e., $\mathcal{V}^-[k] = \{ v_j \in \mathcal{V} \; | \; b_j[k] < 0 \}$.  
During the execution of Algorithm~\ref{alg1_delays}, we have that
$$
\mathcal{V}^-[k+1] \subseteq \mathcal{V}^-[k].
$$
\end{prop}

\begin{proof}
We will first argue that nodes with nonnegative {\em perceived} flow balance at iteration $k$ can never reach negative {\em perceived} flow balance at iteration $k+1$. 
Combining this with the fact that the perceived flow balance of a node is always below its actual flow balance (see Remark~\ref{perceivedSmaller_del}), we establish the proof of the proposition.

Consider a node $v_j$ with a nonnegative perceived flow balance $b^{(p)}_j[k] \geq 0$ (since $b^{(p)}_j[k] \leq b_j[k]$, $\forall \ k \geq 0$ we have that also $b_j[k] \geq 0$). \\
\noindent
We analyze below the following two cases:
\begin{enumerate}
\item at least one neighbor of node $v_j$ has positive perceived flow balance,
\item all neighbors of node $v_j$ have negative or zero perceived flow balance.
\end{enumerate}
In both cases, since $b^{(p)}_j[k] \geq 0$, node $v_j$ will attempt to change the flows of (some of) its incoming and outgoing edges. 
Specifically, node $v_j$ will calculate the desirable flow change $c^{(j)}_{ji}[k]$ ($c^{(j)}_{lj}[k]$) for its incoming (outgoing) edges $(v_j, v_i)$ ($(v_l, v_j)$) where $v_i \in  \mathcal{N}_j^-$ ($v_l \in  \mathcal{N}_j^+$).
Then, 
it transmits the desired flow change $c^{(j)}_{ji}[k]$ ($c^{(j)}_{lj}[k]$) to its incoming (outgoing) edges $(v_j, v_i)$ ($(v_l, v_j)$) where $v_i \in  \mathcal{N}_j^-$ ($v_l \in  \mathcal{N}_j^+$).
In the first case (in which at least one neighbor of node $v_j$ has positive perceived flow balance), we have (i) $b^{(p)}_i[k] > 0$ for some $v_i \in  \mathcal{N}_j^-$, or (ii) $b^{(p)}_l[k] > 0$ for some $v_l \in  \mathcal{N}_j^+$.

For (i) we have that during the iteration $k$ of Algorithm~\ref{alg1_delays}, the incoming edge flows of $v_j$ might change by its in-neighbors (i.e., the flow of an incoming edge $(v_j,v_i)$ might be increased to be equal to $f_{ji}[k+1] = f_{ji}[k] + c_{ji}^{(i)}[k]$ for some $v_i \in \mathcal{N}_j^-$).
In this case, since the transmission of $c_{ji}^{(i)}[k]$ from $v_i$ to $v_j$ might undergo a time delay, we have that $v_j$ sets its outgoing flows to be $ f_{lj}[k+1] = f_{lj}[k] + c_{lj}^{(j)}[k] $ and its {\em perceived} incoming flows to be $ f_{ji}^{(p)}[k+1] = f_{ji}^{(p)}[k] + c_{ji}^{(j)}[k] $. Thus, we have that $b^{(p)}_j[k+1] = 0$. 
[Note however that, after $\tau^{(i)}_{ji}[k]$ time steps (during the iteration $k + \tau^{(i)}_{ji}[k]$) node $v_j$ will receive the desired flow change $c_{ji}^{(i)}[k]$ which was sent from node $v_i$ at time step $k$. 
Then it will update its its {\em perceived} incoming flows to be $ f_{ji}^{(p)}[k+\tau^{(i)}_{ji}[k]+1] = f_{ji}^{(p)}[k+\tau^{(i)}_{ji}[k]] + \overline{c}_{ji}^{(i)}[k+\tau^{(i)}_{ji}[k]] $, which means that $ b^{(p)}_j[k+\tau^{(i)}_{ji}[k]+1] > 0 $.]
As a result, for (i) we have that the nonnegative {\em perceived} flow balance of node $v_j$ at iteration $k$ remains nonnegative at iteration $k+1$.

For (ii) we have that the outgoing edge flows of $v_j$ might change by its out-neighbors $v_l \in  \mathcal{N}_j^+$ and it can be argued in a similar manner. 

In the second case, we have $b^{(p)}_i[k] \leq 0$ for every $v_i \in  \mathcal{N}_j^-$, and $b^{(p)}_l[k] \leq 0$ for every $v_l \in  \mathcal{N}_j^+$. 
This means that the neighbors of $v_j$ will not attempt to change the flows of its incoming and outgoing edges. 
As a result, since $v_j$ will transmit its desired flow changes and then set its outgoing flows to be $ f_{lj}[k+1] = f_{lj}[k] + c_{lj}^{(j)}[k] $ and its {\em perceived} incoming flows to be $ f_{ji}^{(p)}[k+1] = f_{ji}^{(p)}[k] + c_{ji}^{(j)}[k] $, we have that $b^{(p)}_j[k+1] = 0$. 

Overall, we have that during an iteration $k$ of Algorithm~\ref{alg1_delays}, nodes with nonnegative {\em perceived} flow balance can never reach negative {\em perceived} flow balance at iteration $k+1$.
From Remark~\ref{perceivedSmaller_del}, since $b^{(p)}_j[k] \leq b_j[k]$, $\forall \ k \geq 0$, we have that also nodes with nonnegative flow balance can never reach negative flow balance, thus establishing the proof of the proposition.
\end{proof}

\begin{prop}
\label{PROP3_del}
Consider the problem formulation described in Section~\ref{preliminaries}. 
During the execution of Algorithm~\ref{alg1_delays}, it holds that
$$
0 \leq \varepsilon[k+1] \leq \varepsilon[k] \; , \; \; \forall k \geq 0 \; ,
$$
where $\varepsilon[k] \geq 0$ is the total imbalance of the network at iteration~$k$ (see Definition~\ref{defn:totalim}).
\end{prop}

\begin{proof} 
From the third  statement of Proposition~\ref{PROP1_del}, we have $\varepsilon[k+1] = 2 \sum_{v_j \in \mathcal{V}^-[k+1]} \vert b_j[k+1] \vert$ and $\varepsilon[k] = 2 \sum_{v_j \in \mathcal{V}^-[k]} \vert b_j[k] \vert$, whereas from Proposition~\ref{PROP2_del}, we have $\mathcal{V}^-[k+1] \subseteq \mathcal{V}^-[k]$.

Consider a node $v_j \in \mathcal{V}^-[k]$ with flow balance $b_j[k]<0$ (here we have that also $b^{(p)}_j[k]<0$ since $b^{(p)}_j[k] \leq b_j[k]$, $\forall \ k \geq 0$ from Remark~\ref{perceivedSmaller_del}). \\
\noindent
We analyze below the following two cases:
\begin{enumerate}
\item all neighbors of node $v_j$ have negative or zero perceived flow balance,
\item at least one neighbor of node $v_j$ has positive perceived flow balance.
\end{enumerate}
In both cases, node $v_j$ will not make any flow changes on its edges. 
In the first case we have that no node will perform any transmissions and thus, the flow balance of node $v_j$ will not change (i.e., $b_j[k+1] = b_j[k] < 0$ and $b^{(p)}_j[k+1] = b^{(p)}_j[k] < 0$). 
This means that for the first case we have $|b_j[k+1]|=|b_j[k]|$ and thus the contribution of node $v_j$ to $\varepsilon[k+1]$ remains the same as its contribution to $\varepsilon[k]$. 

In the second case, we have (i) $b^{(p)}_i[k] \geq 0$ for some $v_i \in  \mathcal{N}_j^-$, or (ii) $b^{(p)}_l[k] \geq 0$ for some $v_l \in  \mathcal{N}_j^+$. 

For (i) we have that during the iteration $k$ of Algorithm~\ref{alg1_delays}, the incoming edge flows of $v_j$ might change by its in-neighbors (i.e., the flow of an incoming edge $(v_j,v_i)$ might be increased to be equal to $f_{ji}[k+1] = f_{ji}[k] + c_{ji}^{(i)}[k]$ for some $v_i \in \mathcal{N}_j^-$).
In this case (regardless if we have a delay during the transmission of $c_{ji}^{(i)}[k]$ from $v_i$ to $v_j$) we have that $b_j[k+1]$ is either positive or $|b_j[k+1]| < |b_j[k]|$ (i.e., the contribution of node $v_j$ to $\varepsilon[k+1]$ is either zero or smaller than its contribution to $\varepsilon[k]$ using the third statement in Proposition~\ref{PROP1_del}).
For (ii) we have that during iteration $k$ of Algorithm~\ref{alg1_delays}, the out-neighbor of $v_j$ might transmit the desired change amount of the outgoing edge flows to node $v_j$.
In this case, if the transmission of $c_{lj}^{(l)}[k]$ is delayed, then the flow balance of $v_j$ will not change (i.e., $b_j[k+1] = b_j[k] < 0$), but when $v_j$ receives $\overline{c}_{lj}^{(l)}[k+\tau^{(l)}_{lj}[k]]$ then the flow balance of node $v_j$ will satisfy $b_j[k+1] \geq b_j[k]$ and thus $b_j[k+1]$ is either positive or $|b_j[k+1]| < |b_j[k]|$ (i.e., the contribution of node $v_j$ to $\varepsilon[k+1]$ is either zero or smaller than its contribution to $\varepsilon[k]$).
As a result, for both cases, we have $\varepsilon[k+1] \leq \varepsilon[k]$ (using the third statement in Proposition~\ref{PROP1_del}).
\end{proof}

\begin{prop}
\label{PROP4_del}
Consider the problem formulation described in Section~\ref{preliminaries} where the integer circulation conditions in Theorem~\ref{IntTheoremCirc} are satisfied. 
Algorithm~\ref{alg1_delays} balances the flows in the graph in a finite number of steps (i.e., $\exists \ k_0$ so that $\forall k \geq k_0$, $f_{ji}[k_0] = f_{ji}[k]$, $\forall (v_j,v_i) \in \mathcal{E}$, where $0 < l_{ji} \leq f_{ji}[k] \leq u_{ji}$, $\forall (v_j,v_i) \in \mathcal{E}$ and $b_j[k] = b_j[k_0] =0$, $\forall \ v_j \in \mathcal{V}$). 
\end{prop}

\begin{proof} 
During the execution of the proposed distributed balancing algorithm, transmissions on each communication link $(v_l,v_j) \in \mathcal{E}$ are affected by arbitrary (time-varying and inhomogeneous) \textit{bounded} time delays (i.e., $0 \leq \tau^{(j)}_{lj}[k] \leq \overline{\tau}_{lj} \leq \infty$).
This means that the packets transmitted on each link $(v_l,v_j) \in \mathcal{E}$ will eventually reach the corresponding node after a finite number of steps.

By contradiction, suppose Algorithm~\ref{alg1_delays} runs for an infinite number of iterations and its total imbalance remains positive (i.e., $\varepsilon[k]>0$ for all $k$). 
Suppose now that Algorithm~\ref{alg1_delays} runs for an infinite number of iterations and its total imbalance remains positive (i.e., $\varepsilon[k]>0$ for all $k$).
This means that there is always (at each $k$) at least one node with positive flow balance.  
Let $\mathcal{V}^{(+)}[k] = \{ v_j \in \mathcal{V} \; | \; b_j^{(p)}[k] > 0 \}$ be the set of nodes that have positive flow balance at time step $k$.  
Now, let $\mathcal{V}^{(+)}$ denote the set of nodes that have positive flow balance {\em infinitely often}. [Since nodes with positive flow balance can become balanced (but not obtain negative flow balance), this means that nodes in the set $\mathcal{V}^{(+)}$ could become balanced at some iteration, as long as they obtain positive flow balance at later iterations.] 
Let us now denote $\mathcal{V}^{(+p)}$, where $\mathcal{V}^{(+p)} \subseteq \mathcal{V}^{(+)}$, as the set of nodes that have positive {\em perceived} flow balance {\em infinitely often}. 
[Since nodes in $\mathcal{V}^{(+)}$ have positive flow balance {\em infinitely often}, and the delays that affect transmissions on each link are bounded, we have that the nodes that belong in $\mathcal{V}^{(+)}$ will also obtain\footnote{Note that if the incoming (or outgoing) weights of node $v_j$ increase (or decrease) then, due to delays, it will not receive instantly the flow changes (i.e., it will assume no change happened on its incoming (or outgoing) weights). 
However, after a finite number of steps, $v_j$ will receive the flow changes and, by calculating its new perceived flow balance, it will notice that its perceived flow balance increased. 
As a result, if node $v_j$ obtains positive flow balance then, after a finite number of steps (since delays are bounded), it will also obtain positive perceived flow balance.} positive {\em perceived} flow balance {\em infinitely often}].
Also there is at least one node with negative flow balance after an infinite number of iterations (i.e.,  belongs in $\mathcal{V}^{(-)} = \lim_{k \rightarrow \infty} \mathcal{V}^{(-)}[k]$ where $\mathcal{V}^{(-)}[k] = \{ v_j \in \mathcal{V} \; | \; b_j[k] < 0 \}$).
This set is well defined (due to the fact that positively balanced nodes cannot become negatively balanced) and contains at least one node with negative flow balance (otherwise the graph is balanced). 
[Note that, from Remark~\ref{perceivedSmaller_del}, nodes with negative flow balance have also negative perceived flow balance.] 
The above discussion implies that as $k$ goes to infinity, the set of nodes $\mathcal{V}$ can be partitioned into three sets: $\mathcal{V}^{(+)}$, $\mathcal{V}^{(-)}$, and $\mathcal{V} - (\mathcal{V}^{(+)} \cup \mathcal{V}^{(-)})$ (the latter is the set of nodes that remain balanced after a finite number of steps -- and never obtain positive flow balance again). This is shown in Fig.~\ref{FlowConv}.

Since the graph is strongly connected, nodes in the set $\mathcal{V}^{(+)}$ need to be connected to/from nodes in the other two sets. 
This is shown via the dashed edges in Fig.~\ref{FlowConv} (note that the presence of all four types of edges is not necessary, but there has to be at least one edge from a node in $\mathcal{V}^{(+)}$ to a node in one of the two other sets, and at least one edge from a node in one of the two other sets to a node in $\mathcal{V}^{(+)}$). 

Take $\mathcal{S} \subset \mathcal{V}$ to be $\mathcal{V}^{(+)}$ and note that $\mathcal{S}$ has at least one node.
A node $v_j \in \mathcal{V}^{(-)}$ needs to have at least one in-neighbour $v_i$ (such that $v_i \in \mathcal{S}$ or $v_i \in \mathcal{V}-(\mathcal{V}^{(-)} \cup \mathcal{V}^{(+)})$) and at least one out-neighbour $v_l$ (such that $v_l \in \mathcal{S}$ or $v_l \in \mathcal{V}-(\mathcal{V}^{(-)} \cup \mathcal{V}^{(+)})$). 
In case $v_i \in \mathcal{V}-(\mathcal{V}^{(-)} \cup \mathcal{V}^{(+)})$ ($v_l \in \mathcal{V}-(\mathcal{V}^{(-)} \cup \mathcal{V}^{(+)})$), then it is easy to see that $f_{ji}$ ($f_{lj}$) will have to reach the value $f_{ji}=\lceil l_{ji} \rceil$ $f_{lj}$ ($f_{li}=\lfloor u_{ji} \rfloor$). The reason is that if that was not the case, node $v_j$ would attempt infinitely often to decrease the value of $f_{ji}$ (increase the value of $f_{lj}$) implying that node $v_i$ ($v_j$) would obtain positive balance infinitely often, which is a contradiction.

We next consider the case when node $v_j \in \mathcal{V}^{(-)}$ has at least one in-neighbor $v_i$ in $\mathcal{S}$ and/or at least one out-neighbor $v_l$ in $\mathcal{S}$. 
Since $v_i$ ($v_l$) also obtains positive {\em perceived} flow balance infinitely often it will attempt to increase (decrease) the flow $f_{ji}[k]$ (or $f_{lj}^{(p)}[k]$) by $c_{ji}^{(i)}[k]$ (or $c_{lj}^{(l)}[k]$). 
If this increase (decrease) happens, then we have that $b_j[k+\tau^{(i)}_{ji}[k]+1] > b_j[k]$ (or $b_j[k+\tau^{(l)}_{lj}[k]+1] > b_j[k]$) so that $v_j$ either arrives at a nonegative flow balance $b_j[k+\tau^{(i)}_{ji}[k]+1] > 0$ (or $b_j[k+\tau^{(l)}_{lj}[k]+1] > 0$) and after a finite number of steps at a nonnegative {\em perceived} flow balance (which is a contradiction), or $0 > b_j[k+\tau^{(i)}_{ji}[k]+1] > b_j[k]$ ($0 > b_j[k+\tau^{(l)}_{lj}[k]+1] > b_j[k]$) implying\footnote{From the third statement of Proposition~\ref{PROP1_del}, we have $\varepsilon[k+1] = 2 \sum_{v_j \in \mathcal{V}^-[k+1]} \vert b_j[k+1] \vert$ and $\varepsilon[k] = 2 \sum_{v_j \in \mathcal{V}^-[k]} \vert b_j[k] \vert$.} that $\varepsilon[k+\tau^{(i)}_{ji}[k]+1] < \varepsilon[k]$ (or $\varepsilon[k+\tau^{(l)}_{lj}[k]+1] < \varepsilon[k]$), which is also a contradiction because if the integer valued $\varepsilon[k]$ decreases infinitely often it will become zero.

Thus, the only possibility left is that the flows of edges outgoing from nodes in $\mathcal{S}$ cannot increase and the flows of edges incoming to nodes in $\mathcal{S}$ cannot decrease. 
In other words, for $k \geq k_0$ for some large enough $k_0$ we have
\begin{eqnarray*}
f_{ji}[k] = \lceil l_{ji} \rceil & & \forall (v_j, v_i) \in \mathcal{E}^-_\mathcal{S} \; , \\
f_{lj}[k] = \lfloor u_{lj} \rfloor & & \forall (v_l, v_j) \in \mathcal{E}^+_\mathcal{S} \; ,
\end{eqnarray*}
where $\mathcal{E}^-_\mathcal{S}$ and $\mathcal{E}^+_\mathcal{S}$ are defined by \eqref{REALEQinS} and \eqref{REALEQoutS} respectively.  

From the first statement of Proposition~\ref{PROP1_del}, for the set $\mathcal{S}$, we have that $\sum_{v_j \in \mathcal{S}} b_j[k] = \sum_{(v_j, v_i) \in \mathcal{E}^-_\mathcal{S}} f_{ji}[k] - \sum_{(v_l, v_j) \in \mathcal{E}^+_\mathcal{S}} f_{lj}[k]$. Thus, we have
$$
\sum_{(v_j, v_i) \in \mathcal{E}^-_\mathcal{S}} \lceil l_{ji} \rceil - \sum_{(v_l, v_j) \in \mathcal{E}^+_\mathcal{S}} \lfloor u_{lj} \rfloor = \sum_{v_j \in \mathcal{S}} b_j[k] > 0 \; , 
$$
which means that the integer circulation conditions do \textit{not hold} (i.e., a contradiction).

This means that if, after an infinite number of iterations, the total imbalance $\varepsilon$ of Algorithm~\ref{alg1_delays} remains positive, then the integer circulation conditions do \textit{not hold} for the given a strongly connected digraph $\mathcal{G}_d$.

As a result, if the integer circulation conditions do \textit{hold} for the given digraph, then, during the operation of Algorithm~\ref{alg1_delays}, the total imbalance $\varepsilon$ will become equal to zero after a finite number of iterations, and the proposed distributed algorithm will result in a flow-balanced digraph.
\end{proof}

\vspace{-0.3cm}
\begin{figure}[h]
\begin{center}
\includegraphics[width=0.5\columnwidth]{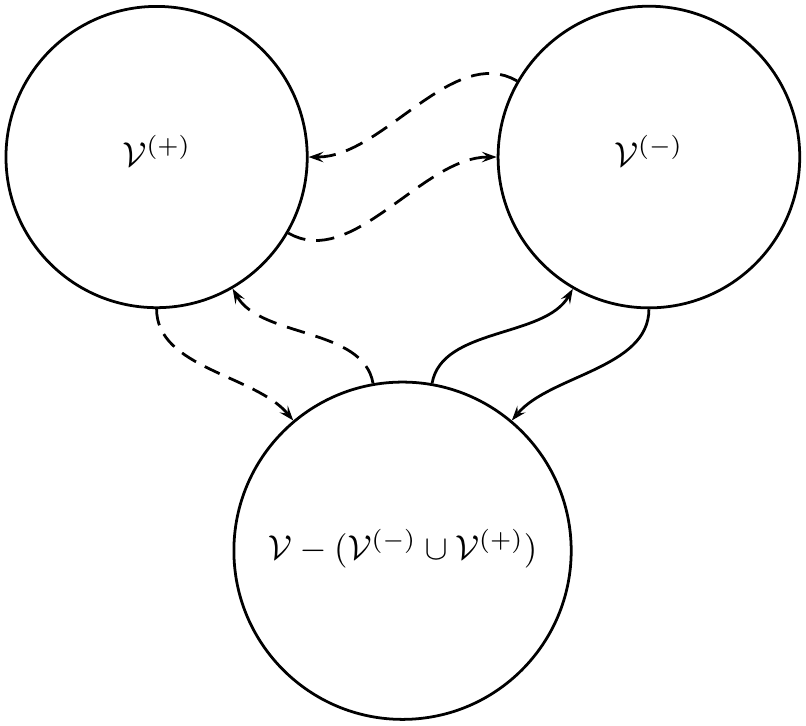}
\caption{Example of digraph where Theorem~\ref{IntTheoremCirc} does not hold for the dashed edges.}
\label{FlowConv}
\end{center}
\end{figure}


\begin{remark}
It is interesting to notice here that in of Algorithm~\ref{alg1_delays} only nodes $v_j$ with positive perceived flow balance $b^{(p)}_j[k] > 0$ execute the proposed protocol (i.e., they calculate the desired flow change for their outgoing and incoming edges) and perform transmissions towards their neighbours. 
Since, during the execution of Algorithm~\ref{alg1_delays}, we have that there exists $k_0$ for which $b_j[k_0] = 0$ for every $v_j \in \mathcal{V}$, this means that once every node $v_j$ reaches weight balancing, it will not perform any other transmission towards its in- and out-neihgbors for time steps $k \geq k_0$. 
Furthermore, a direct extension of the distributed protocol would involve its execution if (and only if), during each time step $k$, each node $v_j$ receives a (possibly delayed) desired flow change $\overline{c}^{(l)}_{lj}[k]$ or $\overline{c}^{(i)}_{ji}[k]$ from its in- and out-neihgbors.
As a result, the proposed distributed protocol can also be implemented in cases where there is need to reduce energy consumption, communication bandwidth, network congestion, and/or processor usage, by considering the use of event-triggered communication and control \cite{dimarogonas2012distributed, 2014:nowzari_cortes}. 
\end{remark}

\section{ROBUST INTEGER FLOW BALANCING ALGORITHM}
\label{algorithm_packetDr}

In this section we consider the case when packet transmissions might undergo unbounded delays (packet drops) and present a distributed flow algorithm that is robust to such events. 
The formal description of the algorithm is provided in Algorithm~\ref{alg1_packetDr}. 
The algorithm is iterative and operates by having, at each iteration, nodes with \textit{positive} (perceived) flow balance attempt to change the integer flows on both their incoming and/or outgoing edges so that they become flow balanced. 
[Note that the operation of Algorithm~\ref{alg1_packetDr} is similar to Algorithm~\ref{alg1_delays} with the main difference being that each node is required to calculate and transmit the {\em desirable flows} (and not the desired change amount) for its incoming and outgoing edges.]
Again, we assume that each node is in charge of assigning the flows on its outgoing edges (i.e., $f_{ji}$ is assigned by node $v_i$; due to possible packet drops the perceived flow $f^{(p)}_{ji}$ on this link by node $v_j$ might be different) which means that each node will know exactly the flows on its outgoing edges but only have access to perceived flows on its incoming edges.

We describe the iterative algorithm operations and we establish that, if the necessary and sufficient integer circulation conditions for the existence of a set of \textit{integer} flows that balance the given digraph are satisfied, the algorithm completes, almost surely, after a finite number of iterations.

\textbf{Initialization.} Same as Algorithm~\ref{alg1_delays}.

\textbf{Iteration.} At each iteration $k \geq 0$, node $v_j$ is aware of the {\em perceived} integer flows on its incoming edges $\{ f^{(p)}_{ji}[k] \; | \: v_i \in \mathcal{N}^-_j \}$ and the (actual) flows on its outgoing edges $\{ f_{lj}[k] \; | \: v_l \in \mathcal{N}^+_j \}$, which allow it to calculate its {\em perceived} flow balance $b^{(p)}_j[k]$ according to Definition~\ref{DEFpercnodebalance}.

\noindent
{\em A. Selecting Desirable Flows:} Each node $v_j$ with positive {\em perceived} flow balance (i.e., $b^{(p)}_j[k] > 0$) attempts to change the flows on its incoming edges $\{ f_{ji}[k] \; | \; v_i \in \mathcal{N}^-_j \}$ and/or outgoing edges $\{ f_{lj}[k] \; | \; v_l \in \mathcal{N}^+_j \}$ in a way that drives its perceived flow balance $b^{(p)}_j[k+1]$ to zero (at least if no other changes are inflicted on the flows).  
No attempt to change flows is made if node $v_j$ has negative or zero perceived flow balance. 
Specifically, node $v_j$ attempts to add $1$ to (or subtract $1$ from) its outgoing (or incoming) integer flows one at a time, according to a predetermined (cyclic) order, until its perceived flow balance becomes zero. 
If an outgoing (incoming) edge has reached its max (min) value (according to the feasible interval on that particular edge), then its flow does not change and node $v_j$ proceeds in changing the next one according to the predetermined order. 
The desired flow by node $v_j$ on edge $(v_j,v_i) \in \mathcal{E}$ at iteration $k$ will be denoted by $f_{ji}^{(j)}[k]$; similarly, the desired flow by node $v_j$ on edge $(v_l,v_j) \in \mathcal{E}$ at iteration $k$ will be denoted by $f_{lj}^{(j)}[k]$.

\noindent
{\em B. Exchanging Desirable Flows:} 
Once the nodes with positive {\em perceived} flow balance calculate the desirable incoming $\{ f^{(j)}_{ji}[k] \; | \; v_i \in \mathcal{N}^-_j \}$ and outgoing $\{ f^{(j)}_{lj}[k] \; | \; v_l \in \mathcal{N}^+_j \}$ flows, they take the following steps in sequence: 

\noindent
1) Node $v_j$ transmits (receives) the calculated desirable flows $f^{(j)}_{ji}[k]$ ($f^{(l)}_{lj}[k]$) to (from) their in- (out-) neighbor $v_i$ ($v_l$).
[Nodes with non-positive perceived balance simply transmit the values $f^{(p)}_{ji}[k]$.]

\noindent
2) If no flow is received from out-neighbor $v_l$ (due to a packet drop), then node $v_j$ assumes that $f^{(l)}_{lj}[k] = f_{lj}[k]$ for the corresponding outgoing edge $(v_l, v_j)$ which suffered a packet drop on the transmission on the reverse link from node $v_l$ to node $v_j$.
Then it calculates its new outgoing flows $ f_{lj}[k+1] = f_{lj}^{(l)}[k] + f_{lj}^{(j)}[k] - f_{lj}[k] $  (projected onto the feasible interval $[l_{lj}, u_{lj}]$) and it transmits them to each corresponding out-neighbor $v_l \in \mathcal{N}^+_j$. 

\noindent
3) It receives the new incoming flows $\{ f^{(p)}_{ji}[k+1] \; | \; v_i \in \mathcal{N}^-_j \}$ from each corresponding in-neighbor. 
If no flow is received then node $v_j$ assumes that $f^{(p)}_{ji}[k+1] = f^{(j)}_{ji}[k]$ for the corresponding incoming edge $(v_j, v_i)$ which suffered a packet drop.
The different flows that the nodes are exchanging (and what happens in the case of a packet drop) are shown in Fig.~\ref{nodes_exchV2}.

\begin{figure}[h]
\begin{center}
\includegraphics[width=0.8\columnwidth]{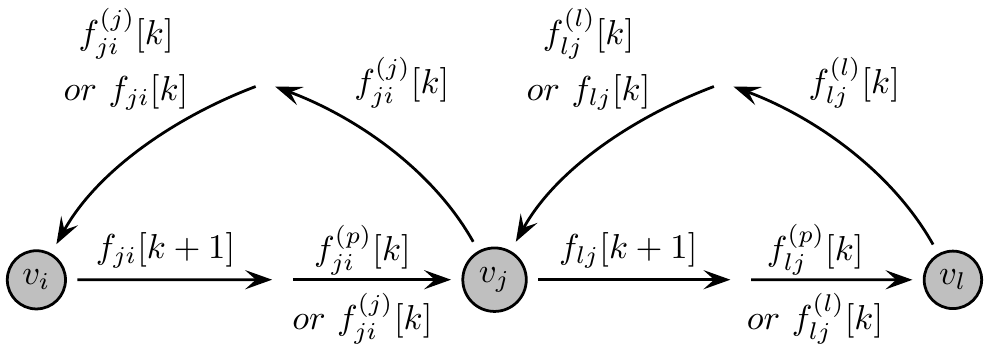}
\caption{Digraph where nodes exchange their desirable flows.}
\label{nodes_exchV2}
\end{center}
\end{figure}

Depending on the possible packet drops that might occur during the exchange of the desirable flows, we have the following four cases:
\begin{enumerate}
\item $f^{(j)}_{ji}[k]$ is dropped,
\item both $f^{(j)}_{ji}[k]$ and $f_{ji}[k+1]$ are dropped,
\item $f_{ji}[k+1]$ is dropped,
\item no packet is dropped.
\end{enumerate}
For the first two cases, the new flow on edge $(v_j,v_i) \in \mathcal{E}$ is taken to be $f_{ji}[k+1] = f^{(i)}_{ji}[k]$ where $l_{ji} \leq f^{(i)}_{ji}[k] \leq u_{ji}$ (the difference in the two cases is that in the second case the perceived value of the flow at node $v_j$ is $f_{ji}^{(p)}[k+1] = f_{ji}^{(j)}[k]$). 

\noindent
For the third and fourth cases, the new flow on edge $(v_j,v_i) \in \mathcal{E}$ is taken to be $f_{ji}[k+1] = [f^{(i)}_{ji}[k+1] + f^{(j)}_{ji}[k+1] - f_{ji}[k]]_{\lceil l_{ji} \rceil}^{\lfloor u_{ji} \rfloor}$ (where $[x]_{\lceil l_{ji} \rceil}^{\lfloor u_{ji} \rfloor}$ denotes the projection onto the interval).
The difference in the two cases is that in the third case we have $f_{ji}^{(p)}[k+1] = f_{ji}^{(j)}[k]$, while in the fourth $f_{ji}^{(p)}[k+1] = f^{(i)}_{ji}[k+1] + f^{(j)}_{ji}[k+1] - f_{ji}[k]$.

\begin{algorithm}
\caption{Distributed Resilient Integer Flow Algorithm}
\textbf{Input} \\ (Inputs are the same as Algorithm~\ref{alg1_delays}). \\
\textbf{Initialization} \\ (Steps~1, 2 are the same as Algorithm~\ref{alg1_delays}).
\\
\textbf{Iteration} \\ For $k=0,1,2,\dots$, each node $v_j \in \mathcal{V}$ does the following:
\\ 1) It computes its \textit{perceived flow balance} as in Definition~\ref{DEFpercnodebalance}.
\\ 2) If  $b^{(p)}_j[k] > 0$, it increases (decreases) by $1$ the integer flows $f_{lj}[k]$ ($f_{ji}^{(p)}[k]$) of its outgoing (incoming) edges $v_l \in \mathcal{N}_j^+$ ($v_i \in \mathcal{N}_j^-$) one at a time, following the predetermined order $P_{lj}$ ($P_{ji}$) until its flow balance becomes zero (if an edge has reached its maximum value, its flow does not change and node $v_j$ proceeds in changing the next one according to the predetermined order in a round-robin fashion). 
Then, it stores the desirable flows on each incoming edge as $f_{ji}^{(j)}[k]$ and each outgoing edge as $f_{lj}^{(j)}[k]$.
\\ 3) If $b^{(p)}_j[k] \leq 0$, it sets $f_{lj}^{(j)}[k] = f_{lj}[k]$ (and $f_{ji}^{(j)}[k] = f_{ji}^{(p)}[k]$) for its outgoing (incoming) edges in $\mathcal{E}_j^+$ ($\mathcal{E}_j^-$).
\\ 4) It transmits the new flow $f_{ji}^{(j)}[k]$ on each incoming edge.
\\ 5) It receives the new flow $f_{lj}^{(l)}[k]$ from each outgoing edge (if no flow was received then it assumes that $f_{lj}^{(l)}[k] = f_{lj}[k]$).
\\ 6) It sets its outgoing flows to be 
$$
f_{lj}[k+1] = f_{lj}^{(l)}[k] + f_{lj}^{(j)}[k] - f_{lj}[k] .
$$
7) It transmits the new flow $f_{lj}[k+1]$ on each outgoing edge.
\\ 8) It receives new flow $f_{ji}^{(p)}[k+1]$ from each incoming edge (if no flow is received then it assumes that $f_{ji}^{(p)}[k+1] = f_{ji}^{(j)}[k]$). 
\\ 9) It repeats (increases $k$ to $k+1$ and goes back to Step~1).
\label{alg1_packetDr}
\end{algorithm}



\begin{remark}\label{perceivedSmaller}
It is important to note here that the total {\em perceived} in-flow $f_j^{-(p)}$ of node $v_j$ might be affected from possible packet drops at Step~7 of Algorithm~\ref{alg1_packetDr}. 
Specifically, if a packet drop occurs; then $v_j$ assumes $f_{ji}^{(p)}[k+1] = f_{ji}^{(j)}[k]$ where $f_{ji}^{(j)}[k] \leq f_{ji}[k+1]$ (since nodes only attempt to make changes on the flows if their perceived flow balance is positive, node $v_i$ can only increase the flow of edge $(v_j, v_i)$).
This means that during the execution of Algorithm~\ref{alg1_packetDr} we have $f_{ji}^{(p)}[k] \leq f_{ji}[k]$ for each edge $(v_j,v_i) \in \mathcal{E}$, at each time step $k$ which means that the perceived balance $b_j^{(p)}[k]$ of node $v_j$ at iteration $k$ is always smaller or equal to its actual balance $b_j[k]$ (i.e., $b_{j}^{(p)}[k] \leq b_{j}[k]$). 
\end{remark}

%
%
%
%

\subsection{Proof of Algorithm Completion}\label{convergence_packetDr}

In this section, we show that, as long as the integer circulation conditions hold, then the total imbalance $\varepsilon[k]$ in Definition~\ref{defn:totalim} goes to zero after a finite number of iterations of Algorithm~\ref{alg1_packetDr}. 
This implies that the flow balance $b_j[k]$ for each node $v_j \in \mathcal{V}$ goes to zero after a finite number of iterations, and thus (from the flow updates in Algorithm~\ref{alg1_packetDr}) the integer flow $f_{ji}[k]$ on each edge $(v_j, v_i) \in \mathcal{E}$ stabilizes to an integer value $f_{ji}^*$ (where $f_{ji}^* \in \mathbb{N}_0$) within the given lower and upper limits, i.e., $1 \leq l_{ji} \leq f^*_{ji} \leq u_{ji}$ for all $(v_j, v_i) \in \mathcal{E}$. 
As in \cite{2016:RikosHadj} we assume that $l_{ji} \geq 1$ for each edge $(v_j, v_i) \in \mathcal{G}$.

\vspace{0.1cm}

\begin{prop}
\label{PROP2_packDr}
Consider the problem formulation described in Section~\ref{preliminaries}. Let $\mathcal{V}^-[k] \subset \mathcal{V}$ be the set of nodes with negative flow balance at iteration $k$, i.e., $\mathcal{V}^-[k] = \{ v_j \in \mathcal{V} \; | \; b_j[k] < 0 \}$.  During the execution of Algorithm~\ref{alg1_packetDr}, we have that
$$
\mathcal{V}^-[k+1] \subseteq \mathcal{V}^-[k].
$$
\end{prop}

\begin{proof}
We will first argue that nodes with nonnegative {\em perceived} flow balance at iteration $k$ can never reach negative perceived flow balance at iteration $k+1$. 
Combining this with the fact that the perceived flow balance of a node is always below its actual flow balance (see Remark~\ref{perceivedSmaller}), we establish the proof of the proposition.

Consider a node $v_j$ with a nonnegative perceived flow balance $b^{(p)}_j[k] \geq 0$ (from Remark~\ref{perceivedSmaller} we have $b_j[k] \geq 0$). \\
\noindent
We analyze below the following two cases:
\begin{enumerate}
\item at least one neighbor of node $v_j$ has positive perceived flow balance,
\item all neighbors of node $v_j$ have negative or zero perceived flow balance.
\end{enumerate}
In both cases, since $b^{(p)}_j[k] \geq 0$, node $v_j$ will attempt to change the flows of (some of) its incoming and outgoing edges. 
Specifically, node $v_j$ calculates the desirable flow $f^{(j)}_{ji}[k]$ ($f^{(j)}_{lj}[k]$) for its incoming (outgoing) edges $(v_j, v_i)$ ($(v_l, v_j)$) where $v_i \in  \mathcal{N}_j^-$ ($v_l \in  \mathcal{N}_j^+$).

In the first case, both in- and out-neighbors ($v_i$ and $v_l$ respectively) of $v_j$ will calculate the desirable flows for their incoming and outgoing edges.
Depending on the possible packet drops that might occur during the transmissions from node $v_i$ to node $v_j$, we consider the following two scenarios:
\begin{enumerate}
\item[a)] no packet is dropped,
\item[b)] at least one packet is dropped.
\end{enumerate}

\noindent
Recall that from the perceptive of node $v_j$ the following transmissions take place: first, node $v_j$ sends $f^{(j)}_{ji}[k]$ to each in-neighbor $v_i \in  \mathcal{N}_j^-$. 
Then it receives $f^{(l)}_{lj}[k]$ from every out-neighbor $v_l \in  \mathcal{N}_j^+$ and finally, once it calculates the new flows $f_{lj}[k+1]$ for its outgoing edges $(v_l, v_j)$ (where $v_l \in  \mathcal{N}_j^+$), it transmits them to every out-neighbor $v_l \in  \mathcal{N}_j^+$.

\noindent
For the first scenario (a), we have 
\begin{eqnarray}
b^{(p)}_j[k+1] & = & \sum_{v_i \in \mathcal{N}_j^-} f^{{(p)}}_{ji}[k+1] - \sum_{v_l \in \mathcal{N}_j^+} f_{lj}[k+1]\label{perc_imb_1} \\
 & = & \sum_{v_i \in \mathcal{N}_j^-} ( f^{(i)}_{ji}[k] + f^{(j)}_{ji}[k] - f_{ji}[k] ) - \nonumber \\
 & & \; \; - \sum_{v_l \in \mathcal{N}_j^+} ( f^{(j)}_{lj}[k] + f^{(l)}_{lj}[k] - f_{lj}[k] ) \; . \nonumber
\end{eqnarray}

\noindent
Since
\begin{equation}\label{perc_imb_2}
\sum_{v_i \in \mathcal{N}_j^-} f^{(j)}_{ji}[k] = \sum_{v_l \in \mathcal{N}_j^+} f^{(j)}_{lj}[k] ,
\end{equation}
(\ref{perc_imb_1}) becomes
\begin{eqnarray}
b^{(p)}_j[k+1] & = & \sum_{v_i \in \mathcal{N}_j^-} ( f^{(i)}_{ji}[k] - f_{ji}[k] ) - \nonumber \\
 & & \; \; - \sum_{v_l \in \mathcal{N}_j^+} ( f^{(l)}_{lj}[k] - f_{lj}[k] )\label{perc_imb_3} \; . 
\end{eqnarray}

\noindent
Also, since $f^{(i)}_{ji}[k] \geq f_{ji}[k]$ and $f^{(l)}_{lj}[k] \leq f_{lj}[k]$, $\forall \ (v_j, v_i), \ (v_l, v_j) \in \mathcal{E}$, we conclude $b^{(p)}_j[k+1] \geq 0, \; \forall v_j \in \mathcal{V}$.
%
%
%
%
As a result we conclude that, for scenario (a), the nonnegative {\em perceived} flow balance of node $v_j$ at iteration $k$ remains nonnegative at iteration $k+1$.

For scenario (b), let us assume (without loss of generality) that $f_{ji}[k+1]$, sent from node $v_i$ to node $v_j$ at Step~7 of the proposed algorithm, suffered a packet drop while all the other transmissions were successful. We have that 
{\small \begin{eqnarray}
b^{(p)}_j[k+1] & = & \sum_{v_{i'} \in \mathcal{N}_j^-} f^{(p)}_{ji'}[k+1] - \sum_{v_l \in \mathcal{N}_j^+} f_{lj}[k+1] \nonumber \\ 
 & = & f_{ji}^{(j)}[k] + \sum_{v_{i'} \in \mathcal{N}_j^- - \{ v_i \} } ( f_{ji'}[k+1] ) + f^{(j)}_{ji}[k] - \nonumber \\
 & & \; \; - \sum_{v_l \in \mathcal{N}_j^+} ( f^{(j)}_{lj}[k] + f^{(l)}_{lj}[k] - f_{lj}[k] ) \; , \nonumber
\end{eqnarray}
}
\noindent
which, in a similar manner, leads to the conclusion that $b^{(p)}_j[k+1] \geq 0, \; \forall v_j \in \mathcal{V}$.
%
%
%
%
Thus, for scenario (b), we conclude that if only the transmission from node $v_i$ to node $v_j$ suffered a packet drop, the nonnegative {\em perceived} flow balance of node $v_j$ at iteration $k$ remains nonnegative at iteration $k+1$.

The remaining scenarios, where multiple transmissions suffer packet drops during the same iteration $k$, as well as the remaining cases, where all neighbors of node $v_j$ have negative or zero perceived flow balance, can be argued in a similar manner.

As a result we have that during an iteration $k$ of Algorithm~\ref{alg1_packetDr}, nodes with nonnegative {\em perceived} flow balance can never reach negative perceived flow balance at iteration $k+1$.
From Remark~\ref{perceivedSmaller}, since $b^{(p)}_j[k] \leq b_j[k]$, $\forall \ k \geq 0$, we have that also nodes with nonnegative flow balance can never reach negative flow balance, thus establishing the proof of the proposition.
\end{proof}

\begin{prop}
\label{PROP3_packDr}
Consider the problem formulation described in Section~\ref{preliminaries}. During the execution of Algorithm~\ref{alg1_packetDr}, it holds that
$$
0 \leq \varepsilon[k+1] \leq \varepsilon[k] \; , \; \; \forall k \geq 0 \; ,
$$
where $\varepsilon[k] \geq 0$ is the total imbalance of the network at iteration~$k$ (see Definition~\ref{defn:totalim}).

\end{prop}

\begin{proof} 
From the third statement of Proposition~\ref{PROP1_del}, we have $\varepsilon[k+1] = 2 \sum_{v_j \in \mathcal{V}^-[k+1]} \vert b_j[k+1] \vert$ and $\varepsilon[k] = 2 \sum_{v_j \in \mathcal{V}^-[k]} \vert b_j[k] \vert$, whereas from Proposition~\ref{PROP2_packDr}, we have $\mathcal{V}^-[k+1] \subseteq \mathcal{V}^-[k]$.

Consider a node $v_j \in \mathcal{V}^-[k]$ with flow balance $b_j[k]<0$ (since $b^{(p)}_j[k] \leq b_j[k]$, we have that also $b^{(p)}_j[k]<0$). \\
\noindent
We analyze below the following two cases:
\begin{enumerate}
\item all neighbors of node $v_j$ have negative or zero perceived flow balance,
\item at least one neighbor of node $v_j$ has positive perceived flow balance.
\end{enumerate}
In both cases, node $v_j$ will not make any flow changes on its edges. 
In the first case, the flow balance of node $v_j$ will not change (i.e., $b_j[k+1] = b_j[k] < 0$). 
In the second case, we have (i) $b^{(p)}_i[k] \geq 0$ for some $v_i \in  \mathcal{N}_j^-$, or (ii) $b^{(p)}_l[k] \geq 0$ for some $v_l \in  \mathcal{N}_j^+$.

For (i) we have that during the iteration $k$ of Algorithm~\ref{alg1_packetDr}, the incoming edge flows of $v_j$ might change by its in-neighbors (i.e., the flow of an incoming edge $(v_j,v_i)$ might be increased to be equal to $f_{ji}[k+1] = f_{ji}^{(i)}[k]$ for some $v_i \in \mathcal{N}_j^-$).
In this case (regardless if we have a packet drop during the transmission of $f_{ji}[k+1]$ from $v_i$ to $v_j$) we have that $b_j[k+1]$ is either positive or $|b_j[k+1]| < |b_j[k]|$ (i.e., the contribution of node $v_j$ to $\varepsilon[k+1]$ is either zero or smaller than its contribution to $\varepsilon[k]$ using the third statement in Proposition~\ref{PROP1_del}).
For (ii) we have that during iteration $k$ of Algorithm~\ref{alg1_packetDr}, the out-neighbor of $v_j$ might transmit the new outgoing edge flows to node $v_j$ (i.e., $v_j$ might receive the new $f_{lj}^{(l)}[k]$ from some $v_l \in \mathcal{N}_j^+$).
In this case, if $f_{lj}^{(l)}[k]$ suffers a packet drop, the flow balance of $v_j$ will not change (i.e., $|b_j[k+1]|=|b_j[k]|$) and thus the contribution of node $v_j$ to $\varepsilon[k+1]$ remains the same as its contribution to $\varepsilon[k]$.
If $f_{lj}^{(l)}[k]$ is transmitted successfully we have that $b_j[k+1]$ is either positive or $|b_j[k+1]| < |b_j[k]|$ (i.e., the contribution of node $v_j$ to $\varepsilon[k+1]$ is either zero or smaller than its contribution to $\varepsilon[k]$ using the third statement in Proposition~\ref{PROP1_del}).
As a result, for both cases, we have $\varepsilon[k+1] \leq \varepsilon[k]$ (using the third statement in Proposition~\ref{PROP1_del}).
\end{proof}

\begin{prop}
\label{PROP4_packDr}
Consider the problem formulation described in Section~\ref{preliminaries} where the integer circulation conditions in Theorem~\ref{IntTheoremCirc} are satisfied. Algorithm~\ref{alg1_packetDr} balances the flows in the graph in a finite number of steps, with probability one (i.e., $\exists \ k_0$ so that almost surely $\forall k \geq k_0$, $f_{ji}[k_0] = f_{ji}[k]$, $\forall (v_j,v_i) \in \mathcal{E}$ and $b_j[k] = b_j[k_0] =0$, $\forall v_j \in \mathcal{V}$). 
\end{prop}

\begin{proof} 
By contradiction, suppose Algorithm~\ref{alg1_packetDr} runs for an infinite number of iterations and its total imbalance remains positive (i.e., $\varepsilon[k]>0$ for all $k$). 
During the execution of the proposed distributed balancing algorithm, packets containing information are dropped with probability $q_{lj} <1$ for each communication link $(v_l,v_j) \in \mathcal{E}$ (we assume independence between packet drops at different time steps and different links and link directions).
During transmissions on link $(v_l, v_j)$, we have that at each transmission, a packet goes through with probability $1 - q_{lj} > 0$.
Thus, if we consider $k_{lj}$ consecutive uses of link $(v_l, v_j)$, the probability that at least one packet will go through is $1 - q_{lj}^{k_{lj}}$, which will be arbitrarily close to $1$ for a sufficiently large $k_{lj}$.
\noindent
Specifically, for any (arbitrarily small) $\epsilon > 0$, we can choose 
$$
k_{lj} = \left \lceil \frac{\log \epsilon}{\log q_{lj}} \right \rceil, 
$$
to ensure that each transmission goes through by $k_{lj}$ steps with probability $1 - \epsilon$. 

Suppose now that Algorithm~\ref{alg1_packetDr} runs for an infinite number of iterations (where infinite successful packet transmissions occurred on each link $(v_l, v_j)$, for a sufficiently large $k_{lj}$) and its total imbalance remains positive (i.e., $\varepsilon[k]>0$ for all $k$).
This means that there is always (at each $k$) at least one node with positive flow balance and thus the proof of this Proposition becomes identical to the proof of Proposition~\ref{PROP4_del}. 
\end{proof}


%
%
%
%
\section{SIMULATION RESULTS}\label{results}

In this section, we present simulation results and comparisons for the proposed distributed algorithms. 
Specifically, we illustrate the behavior of the proposed distributed algorithms for the following two scenarios: 
(i) the scenario where Algorithm~\ref{alg1_delays} operates in a randomly created graph of $20$ nodes where for every communication link $(v_j,v_i) \in \mathcal{E}$ there are bounded transmission delays $0 < \tau_{lj} < \overline{\tau}$, at each iteration $k$, where $\overline{\tau} = 10$ (we choose the delays randomly, independently between different links and link directions, but keep in mind that the profile of the delays could be anything as long as they are bounded),
(ii) the scenario where Algorithm~\ref{alg1_packetDr} operates in a randomly created graph of $20$ nodes where for every communication link $(v_j,v_i) \in \mathcal{E}$ there are packet drops with equal probability $q$, at each iteration $k$, where $0 \leq q < 1$ (independently between different links and link directions).
Note that the the integer circulation conditions hold for both of the following scenarios.

In Fig.~\ref{working20_del} we show the operation of Algorithm~\ref{alg1_delays} in a randomly created graph of $20$ nodes where for every communication link $(v_j,v_i) \in \mathcal{E}$ there are bounded transmission delays $0 < \tau_{lj} < \overline{\tau}$, at each iteration $k$, where $\overline{\tau} = 10$.
In the top plot, we show the \textit{absolute imbalance} $\varepsilon = \sum_{j=1}^{n} \vert b_j \vert$, $\forall v_j \in \mathcal{V}$ (blue line) and the perceived total imbalance $\varepsilon^{(p)} = \sum_{j=1}^{n} \vert b^{(p)}_j \vert$ (red line) against the number of iterations $k$.
In the In the bottom plot, we show \textit{nodes balances} $b_j[k]$ (as defined in Definition~\ref{DEFnodebalance}) as a function of the number of iterations $k$ for the distributed algorithm. 
As expected, the plots verify that the absolute imbalance $\varepsilon$ becomes equal to zero after a finite number of iterations, which means that Algorithm~\ref{alg1_delays} is able to obtain a set of integer flows that balance the corresponding digraph after a finite number of iterations in the presence of bounded transmission delays $0 < \tau_{lj} < \overline{\tau}$, where $\overline{\tau} = 10$, on each link $(v_l,v_j)\in \mathcal{E}$.

\begin{figure}[h]
\begin{center}
\includegraphics[width=0.8\columnwidth]{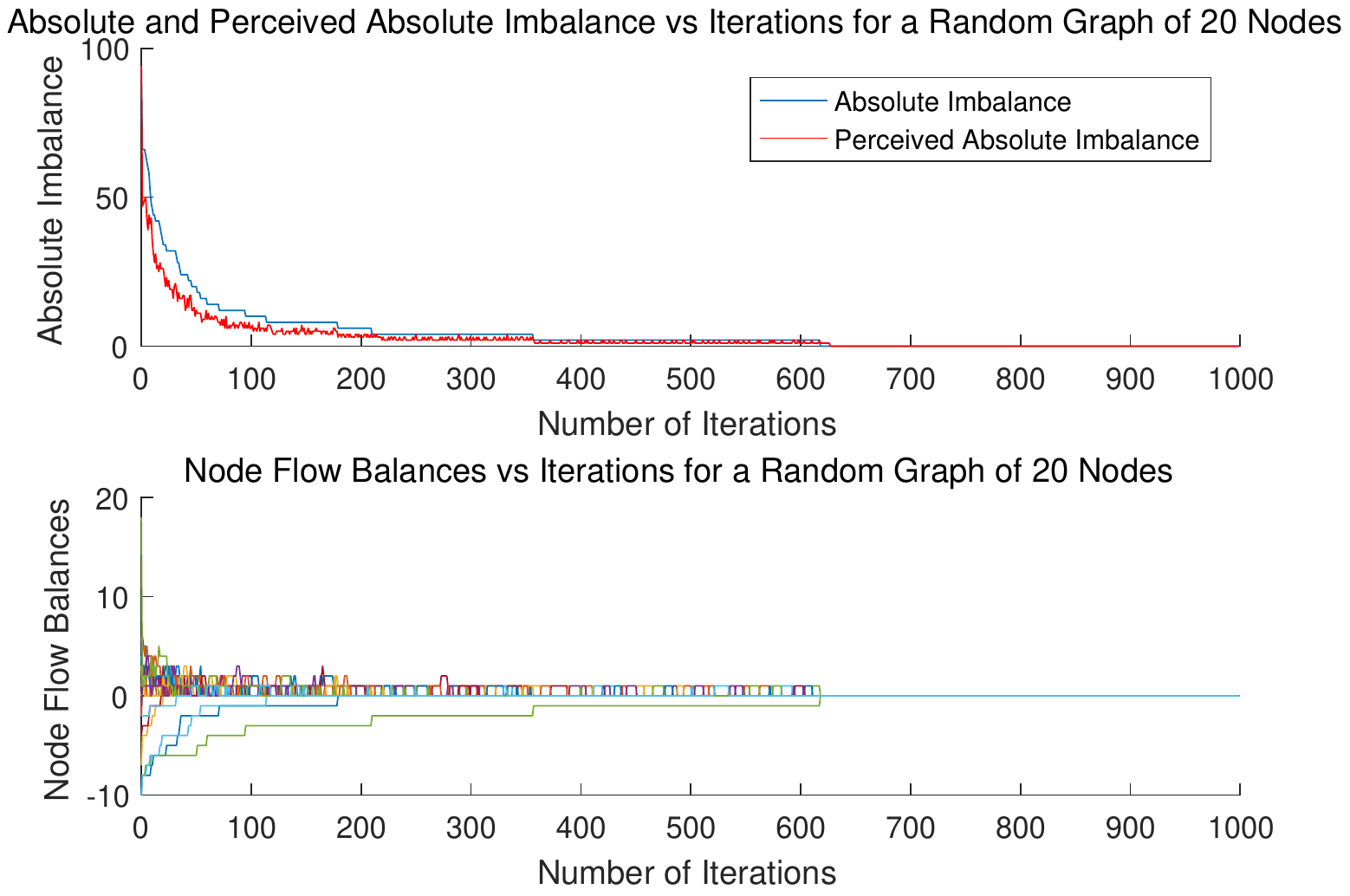}
\caption{Execution of Algorithm~\ref{alg1_delays} for the case when the integer circulation conditions hold for a random graph of $20$ nodes with transmission delays $0 < \tau_{lj} < \overline{\tau}$ where $\overline{\tau} = 10$. \emph{Top figure:} Total (absolute) imbalance $\varepsilon[k]$ (blue line) and perceived total imbalance $\varepsilon^{(p)}[k]$ (red line) plotted against number of iterations. \emph{Bottom figure:} Node flow balances $b_j[k]$ plotted against number of iterations.}
\label{working20_del}
\end{center}
\end{figure}



In Fig.~\ref{working20_packetDr} we show the operation of Algorithm~\ref{alg1_packetDr} for the same case as Fig.~\ref{working20_del}.
The plot suggests that Algorithm~\ref{alg1_packetDr} is able to obtain a set of integer flows that balance the corresponding digraph after a finite number of iterations in the presence of packet dropping links with probability $q = 0.8$, on each link $(v_l,v_j)\in \mathcal{E}$.

\begin{figure}[h]
\begin{center}
\includegraphics[width=0.8\columnwidth]{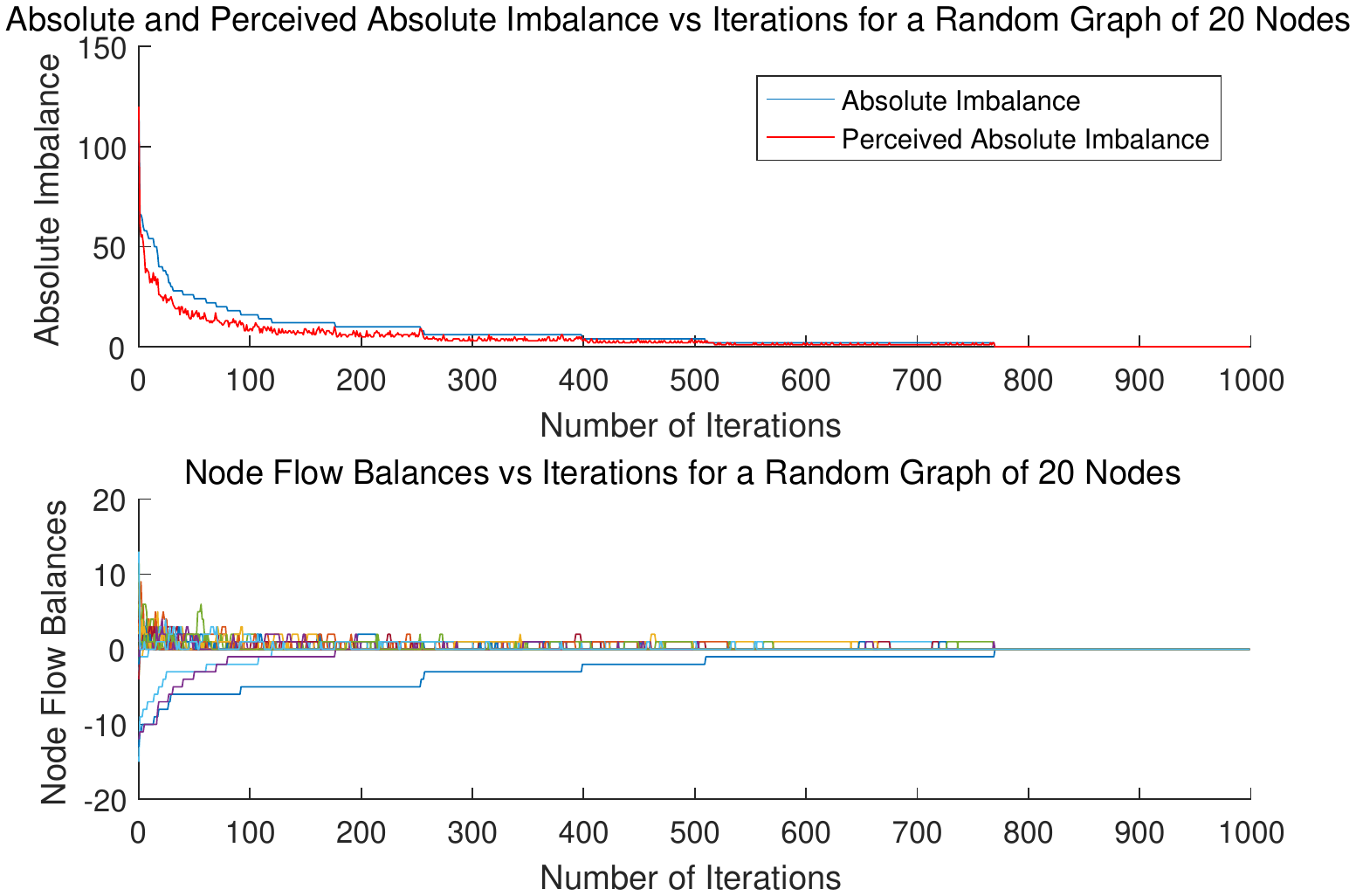}
\caption{Execution of Algorithm~\ref{alg1_packetDr} for the case when the integer circulation conditions hold for a random graph of $20$ nodes with packet drop probability $q_{ji} = 0.8$. \emph{Top figure:} Total (absolute) imbalance $\varepsilon[k]$ (blue line) and Perceived Total Imbalance $\varepsilon^{(p)}[k]$ (red line) plotted against number of iterations. \emph{Bottom figure:} Node flow balances $b_j[k]$ plotted against number of iterations.}
\label{working20_packetDr}
\end{center}
\end{figure}

%
%
%
%
\section{CONCLUSIONS}  \label{future}

We have considered the integer flow/weight balancing problem in a distributed multi-component system whose interconnection topology forms a strongly connected digraph, under the constraint that each edge flow lies within a given interval and communication between links may suffer bounded delays or unbounded delays (packet drops).
We have presented two distributed algorithms, which achieve integer flow-balancing in a multi-component system in the presence of lower and upper limit constraints on the edge flows/weights, and we analysed their functionality and established their correctness in the presence of transmission delays and packet dropping links. 
We also demonstrated their operation, performance, and advantages via various simulations. 

In the future, we plan to characterize the number of steps required for the proposed algorithm to terminate, calculate its computational cost, and compare it with existing algorithms on flow/weight balancing over networks.
We also plan to apply these techniques to distributed averaging consensus problems that are subject to quantized communication.

\bibliographystyle{IEEEtran}
\bibliography{bibliografia_consensus}

\vspace{-0.3cm}

\begin{IEEEbiography}[{\includegraphics[width=1in,height=1.25in,clip,keepaspectratio]{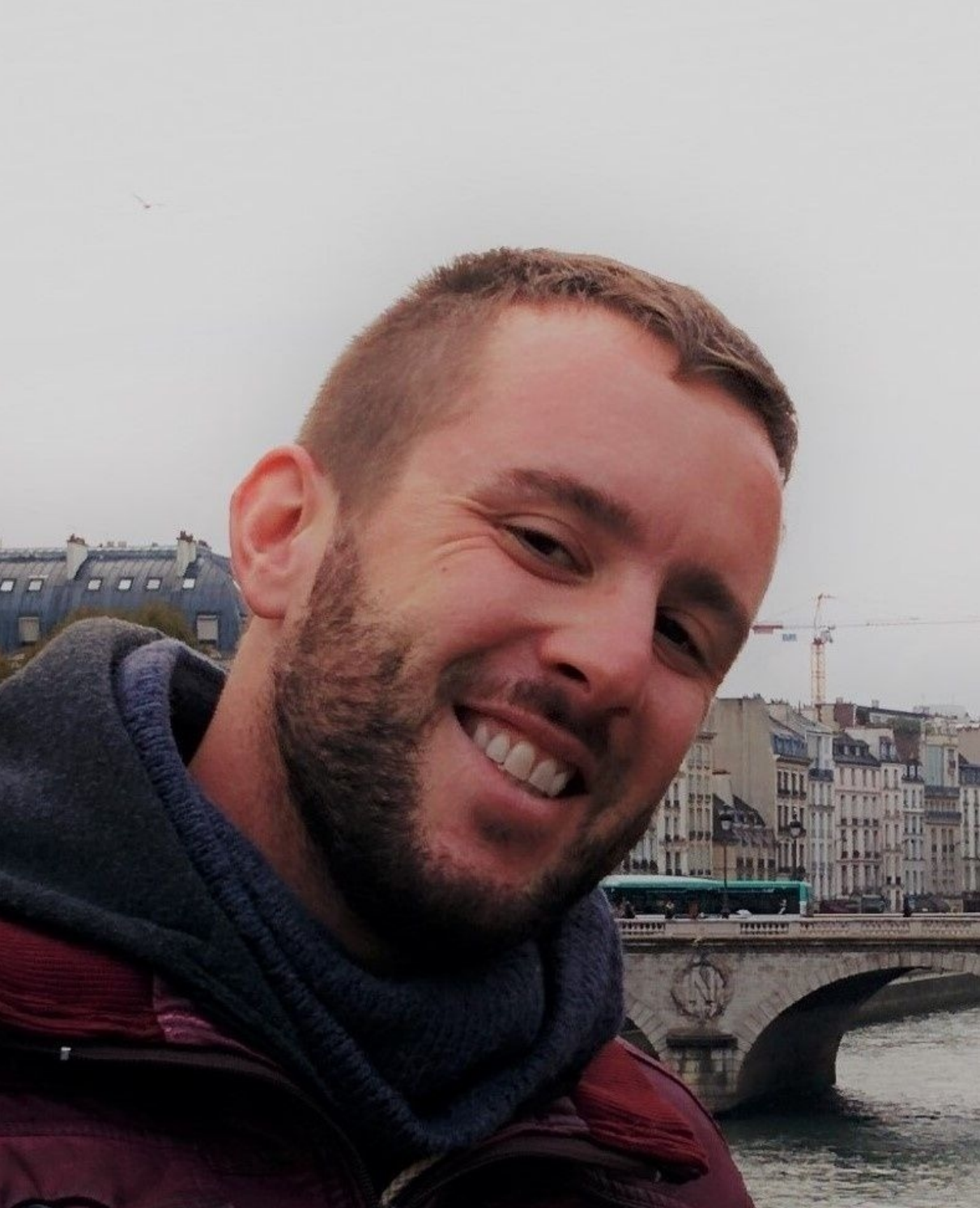}}]{Apostolos I. Rikos} received the B.Sc., M.Sc and Ph.D. degrees in Electrical Engineering from the Department of Electrical and Computer Engineering, University of Cyprus in 2010, 2012 and 2018 respectively.

Since November 2018, he has been working as a Special Scientist in KIOS Research and Innovation Centre of Excellence. His research interests are in the area of distributed systems, coordination and control of networks of autonomous agents, sensor networks and graph theory.
\end{IEEEbiography}

\vspace{-0.2cm}

\begin{IEEEbiography}[{\includegraphics[width=1in,height=1.25in,clip,keepaspectratio]{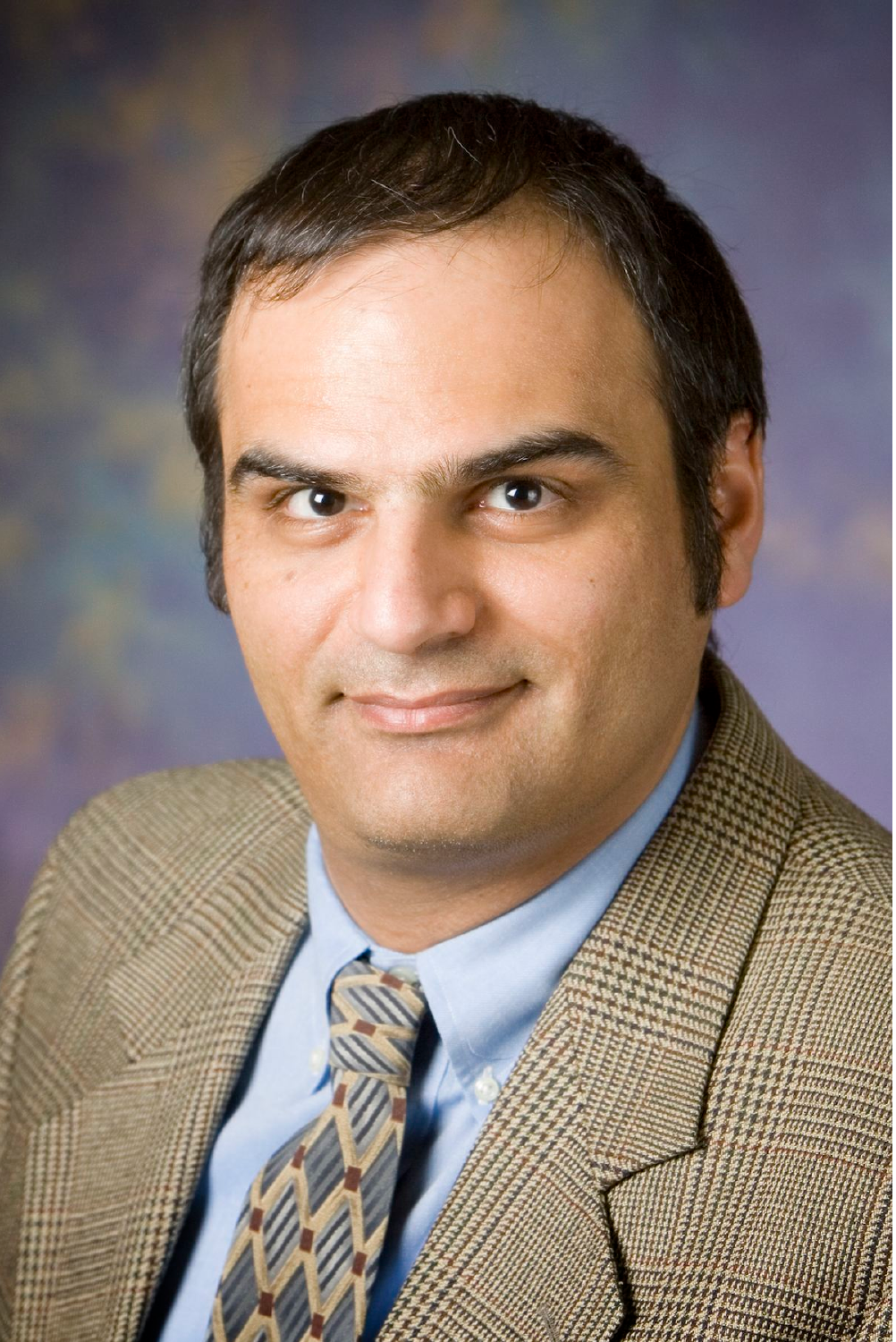}}]{Christoforos N. Hadjicostis}
(M'99, SM'05) received the S.B. degrees in electrical engineering, computer science and engineering, and in mathematics, the M.Eng. degree in electrical engineering and computer science in 1995, and the Ph.D. degree in electrical engineering and computer science in 1999, all from Massachusetts Institute of Technology, Cambridge. In 1999, he joined the Faculty at the University of Illinois at Urbana--Champaign, where he served as Assistant and then Associate Professor with the Department of Electrical and Computer Engineering, the Coordinated Science Laboratory, and the Information Trust Institute. Since 2007, he has been with the Department of Electrical and Computer Engineering, University of Cyprus, where he is currently Professor and Dean of Engineering. His research focuses on fault diagnosis and tolerance in distributed dynamic systems, error control coding, monitoring, diagnosis and control of large-scale discrete-event systems, and applications to network security, anomaly detection, energy distribution systems, medical diagnosis, biosequencing, and genetic regulatory models. He currently serves as Associate Editor of IEEE Transactions on Automatic Control, and IEEE Transactions on Automation Science and Engineering; he has also served as Associate Editor of IEEE Transactions on Control Systems Technology, and IEEE Transactions on Circuits and Systems~I.
\end{IEEEbiography}

\end{document}